\documentclass{article}
\usepackage{graphicx}
\usepackage{wrapfig}
\usepackage{subcaption}
\usepackage{seqsplit}
\usepackage{verbatim}
\usepackage{amsmath,amssymb,amsthm}
\usepackage[utf8]{inputenc}
\usepackage{color}
\usepackage{url}
\usepackage{fullpage}
\usepackage{natbib}
\usepackage{hyperref}
\usepackage[T1]{fontenc}

\newcommand{\ain}{a_{\textrm{in}}}
\newcommand{\aout}{a_{\textrm{out}}}

\newcommand{\pin}{p_{\textrm{in}}}
\newcommand{\pout}{p_{\textrm{out}}}
\newcommand{\win}{w_{\textrm{in}}}
\newcommand{\wout}{w_{\textrm{out}}}
\newcommand{\notP}{\overline{P}}
\newcommand{\tb}{\tilde{b}}
\newcommand{\bfx}{\mathbf{x}}

\DeclareMathOperator*{\argmin}{arg\,min}

\graphicspath{ {./simulations_results_cut_bw}}

\newtheorem{theorem}{Theorem}%  meant for continuous numbers
\newtheorem{lemma}{Lemma}%
\newtheorem{proposition}{Proposition}% 
\newtheorem{corollary}{Corollary}

\theoremstyle{remark}%
\newtheorem{remark}{Remark}%

\theoremstyle{definition}%
\newtheorem{definition}{Definition}%
\newtheorem{observation}{Observation}%

\begin{document}
\title{On the Impact of Social Media Recommendations on Opinion Consensus}


\author{Vincenzo Auletta \and Antonio Coppola \and Diodato Ferraioli\\
Universit\`a degli Studi di Salerno\\
Email: {\tt \{auletta,ancoppola,dferraioli\}@unisa.it}}
\date{}

% \author*[1]{\fnm{Antonio} \sur{Coppola}}\email{ancoppola@unisa.it}
% % \equalcont{These authors contributed equally to this work.}
% 
% \author*[1]{\fnm{Diodato} \sur{Ferraioli}}\email{dferraioli@unisa.it}
% % \equalcont{These authors contributed equally to this work.}
% 
% \affil[1]{\orgdiv{DIEM}, \orgname{Universit\`{a} degli Studi di Salerno}, \orgaddress{\street{via Giovanni Paolo II, 132}, \city{Fisciano}, \postcode{84084}, \country{Italy}}}

\maketitle


\abstract{
We consider a discrete opinion formation problem in a setting where agents are influenced by both information diffused by their social relations and from recommendations received directly from the social media manager. We study how the ``strength'' of the influence of the social media and the homophily ratio affect the probability of the agents of reaching a consensus and how these factors can determine the type of consensus reached. 

In a simple 2-symmetric block model we prove that agents converge either to a consensus or to a persistent disagreement. In particular, we show that when the homophily ratio is large, the social media has a very low capacity of determining the outcome of the opinion dynamics. On the other hand, when the homophily ratio is low, the social media influence can have an important role on the dynamics, either by making harder to reach a consensus or inducing it on extreme opinions. 

Finally, in order to extend our analysis to more general and realistic settings we give some experimental evidences that our results still hold on general networks.

% \blfootnote{Partially supported by GNCS-INdAM and by the Italian MIUR PRIN 2017 Project ALGADIMAR “Algorithms, Games, and Digital Markets”. A previous version of this paper was previously published at the conference AIxIA 2021}
}

\section{Introduction}
Over the last years, we witnessed a rapid rise of the role of online social networking platforms, such as Facebook or Twitter, in our life. As a consequence, individuals increasingly rely on these social platforms to get news and form their opinions. E.g., according to Pew Research Center survey in 2018 \citep{shearer} 68\% of American adults get news on social media, a significant rise from 49\% of 2012. Moreover, it has been observed that social media may have a relevant effect in many real-world critical settings, such as in electoral campaigns \citep{ioana,fujiwara}. For example, some studies showed that the social media may lead to extremism \citep{benigni} and polarization in individuals' opinions \citep{allcott}.

Hence, it urges to understand how the social media may affect the process of opinion formation of their users. To this aim, several models have been introduced to describe how the opinions of agents evolve under the effect of the social influence. The first such model, due to \cite{DeGroot}, states that each agent adopts an opinion that averages among the ones of individuals which she interacts with. One of the most relevant extensions of this model is, undoubtedly, the dynamics described by \cite{FJ} (see also, the work of \cite{bindel}), that limits the effects of social influence by holding agents close to their original ideology. These models assume that opinions may take values in a continuous space, and agents may express any value in this space. However, in several real settings, i.e., electoral contexts, the number of alternatives around which opinions should converge are limited. Moreover, even if opinions can take values that do not match any alternative, these cannot be expressed due to the limitedness of the options according to which opinions are expressed (e.g., polls, finite-precision ranks, etc.). 
For these reasons, continuous models turn out to be scarcely representative in some settings, and discrete versions of these models have been proposed in which agents' opinions must belong to a discrete set \cite{chierichetti,goldberg}.

However, in several settings it is not sufficient to take into account only the social influence among agents', but we have also to understand how the social media may influence the opinion formation process, and whether and how it is necessary to mitigate in some way the effects it provokes.

There has been recently an increasing interest on these questions. In particular, most of the recent literature in the social choice area focuses on the opportunity for the social media to manipulate the opinion formation process in order to support a target opinion. Different forms of manipulations have been studied, such as seeding, edge addition/deletion, and alteration of the order of changes (see Related Works section for more details).

In this work, we deviate from this approach, and we do not consider the social media as a manipulator. That is, the social media does not have a target that should be promoted, but it only acts as a platform for sharing information. However, social media's goal is to maximize the activity of the agents on the platform and it implements policies about which, when, and to whom information are shared, in order to maximize engagement of users to their service.
While the actual implementation of these policies is private, 
% and it may change from one media to another, 
it is evident that users are more likely to be exposed to information closer to their own opinion \citep{bakshy,levy}. \cite{halberstam} have proved that agents have larger probability of interacting (by viewing, liking, or re-sharing) with this kind of information, witnessing in this way their major engagement with the social media.
% While
% these algorithms are proprietary, there is evidence that they are more likely to expose individuals to
% news matching their ideology (Bakshy et al. 2015, Levy 2020); indeed individuals are more likely to
% share posts matching their ideology suggesting algorithms that promote such posts increase engage-
% ment (Halberstam and Knight 2016). The objective of this paper is to investigate the effect of such
% platform decisions on opinion dynamics.

In this paper we want to answer the following question: how much a social media implementing these policies can influence the opinion formation process? This problem has been recently addressed by \cite{immorlica} in the context of continuous opinion formation processes. Their answer depends on the strength of the influence of the social media platform on individuals: if this is high, then agents' opinions tend to extremes; if low, agents' opinion tend to converge; in the middle, instead, some non-extreme disagreement can occur.

However, the continuous approach adopted by \cite{immorlica} does not fit with many real world critical contexts, such as in voting, in which we usually have a discrete and limited number of candidates around which opinions should converge. For this reason, in this work, we will depart from the work of \cite{immorlica}, by focusing on the discrete opinion formation process, as defined by \cite{goldberg}.

\paragraph*{Our Contribution} 
In this work, we evaluate the impact of social media recommendations with respect to their influence on the ability of users to reach a consensus. Indeed, the likelihood that a consensus is reached has been widely adopted for comparing different opinion models, and for evaluating the impact that variations on the model may have on opinion formation \citep{DeGroot,HK,fanelli} (see Related Works for more details). Note also that consensus is a required goal in many practical settings: from the analysis of collective behaviour of flocks and swarms \citep{olfati2006flocking,savkin2004coordinated}, to sensor fusion \citep{olfati2005consensus}, to formation control for multi-robot systems \citep{egerstedt2001formation,tanner2004leader,lin2005necessary}.

In this work we first focus on a very simple class of networks, namely \emph{symmetric two-block model}, already analyzed by \cite{immorlica}, in which agents are separated in two components, and agents from the same component have the same initial opinion and receive the same influence from individuals inside and outside their component. Despite of the simplicity of this network, it highlights a very important difference with respect to the results given by \cite{immorlica}: namely, the impact of the social media not only depends on the strength of the social media influence, but also on the \emph{homophily ratio}, that is how much individuals weight their similars with compared to others. This measure has been often showed to be a key attribute in opinion formation dynamics (see, e.g., \citep{dandekar}). Hence, our results show a better alignment with respect to the previous literature than the one given by \cite{immorlica}.

Specifically, we will show that whenever the strength of the social media influence is large, consensus is essentially impossible to achieve whenever the initial opinions of the two groups are far from each other. Interestingly, for these initial opinions, consensus is also impossible to achieve when the homophily ratio is large, but the strength of the social media is very small. We also show how the chance of reaching a consensus changes with respect to how extreme are the initial opinions in the two groups. Finally, when initial opinions are instead close to each other, we show that consensus is always possible, but the likelihood of reaching a consensus increases when the homophily ratio is large or the strength of the social media is low.

We conjecture that these findings hold not only for the simple symmetric two-block model, but also for more complex networks whenever initial opinions can be partitioned in two macro-blocks. As an evidence of this conjecture, we provide a massive set of experiments both on synthetic and on real networks: all our experiments show that the dynamics essentially follows the behaviour prescribed by results on the symmetric two-block model as the strength of the social media, the homophily ratio, and the value of initial opinions change.

\paragraph*{Related Works}
% \smallskip \noindent \emph{Related Works.}
Several extensions have been recently proposed to the seminal models by DeGroot and by Friedkin and Johnsen (and their discrete counterparts), by considering only limited interaction by agents \citep{fotakis1,fotakis2}, or an evolving environment \citep{HK,bahawalkar,bilo,ventre,fanelli}, or both repulsive and attractive interaction \citep{ijcai16,acar}. Despite their larger adherence with many real world aspects, however none of these variants has received the same level of interest as the models by DeGroot and by Friedkin and Johnsen. Moreover, the simplicity of the latter models allows a more clear analysis of the influence of social media, by untying it from the complexities of the former models.

Consensus in opinion formation has been object of intense research since the seminal work of \cite{DeGroot}. Indeed, most works aim to evaluate opinion formation models based on their ability to reach a consensus \citep{HK,fanelli}. Many other works try to characterize the parameters that enable a given dynamics to reach consensus \citep{feldman,mossel,greco}. In this work we pursue
% an approach that borrows idea from both these directions:
both approaches:
on one side, we investigate on how the social media recommendations may vary the probability that a consensus is reached; on the other side, we identify the settings, in terms of homophily ratio, strength of the social media influence, and initial agents' opinions, where the probability of consensus is larger.

The study of the influence of a (non-manipulating) social media on the opinion formation process has been initiated by \cite{immorlica}, where, as described above, the focus is on continuous opinions, while we here consider discrete opinions.

Many works instead focus on manipulation of the opinion formation process in social networks, in particular in the framework of election manipulation.
The first and most studied manipulation technique is \emph{seeding}, that consists in selecting a set of sources of news from which to start a successful viral campaign in favour of a designed candidate or against her competitors \citep{KTT,wilder,coro,coro2,castiglioni,elkind}.
% This kind of manipulation has been first studied in the setting of only two candidates and simple opinion dynamics that allow to revise their opinion only once: for this setting, it has been showed that, given a budget $B$, it is possible to efficiently compute a set of seeds that returns a constant approximation of the maximum expected number of voters for the designed candidates \cite{KTT}. This result has been then extended to consider multiple candidates, and it has been showed that it is possible to efficiently approximate the maximum expected margin of victory of the designed candidate \cite{wilder,coro}, but it is hard to provide an approximation of its maximum probability of victory \cite{wilder}. The problem has been then studied in more complex settings involving incomplete knowledge of the manipulator \cite{coro2}, more complex viral campaigns \cite{castiglioni}, and more complex dynamics allowing for more than one opinion revision \cite{elkind}: in all these setting it has been proved that is usually hard to provide an efficient approximation of the best manipulation strategy.
% 
Another kind of manipulation that received large interest consists in adding or deleting links \citep{sina,savarese,elkind,castiglioni}: these may be implemented by social media by hiding the content of a ``friend'' or ``neighbour'' in the social network, or promoting the content of non-friends (e.g., as advertised content or through the mechanism of friend suggestion).
% This kind of manipulation has been implemented first against strategic voters, that, i.e., choose the candidate to vote in order to maximize the chances that their preferred candidates will win according to their knowledge limited to their friends. In this setting, it has been showed that optimizing the manipulation is essentially hard and heuristic results have been provided \cite{sina,savarese}. Next, the manipulation has been applied also in more classic opinion formation settings, in which agents' opinion are simply influenced by their friends, by showing that computing a good approximation is usually hard \cite{elkind,castiglioni}
% 
% Yet 
A last kind of manipulation that recently received a lot of interest consists in guiding the dynamics by influencing the order in which agents are prompted to update their opinion (e.g., by delaying the visualization of a news) so that they will update only when there are enough friends to push them towards the desired candidate \citep{wine15,wine17,aamas,greco,greco2}.
% : it has been showed that, when there are only two available candidates and the opinion dynamics is essentially the one described in \cite{goldberg}, then it is possible to lead the minority candidate to become majority (even in variants of the dynamics in \cite{goldberg}) \cite{wine15,wine17,aamas}, and a candidate with only half voters to reach the consensus \cite{greco}. However, the problems become harder when multiple candidates are involved (even by consider simpler opinion dynamics) \cite{greco2}.
% 
We note that, as described above, our work differs from all these works, since we are not considering a social media operating with the goal of promoting a specific candidate.

\section{The Model}
\label{secTheModel}
We consider $n \geq 2$ agents whose relationships are embedded into a social network modelled as an undirected weighted graph $G = (V, E, w)$, where each vertex of the graph represents an agent. Each agent $i$ keeps an opinion $x_i^0 \in \Theta = \{-1, -1 + \delta, \ldots, -\delta, 0, \delta, \ldots, 1-\delta, 1\}$ for some $0 < \delta \leq \frac{1}{2}$. We will sometimes denote $\delta$ as the \emph{discretization factor} of $\Theta$.
One may think about $\Theta$ as the set of alternatives (e.g., candidates to an election) on which agents' opinions need to converge: note that we are assuming that there is no way for an agent to express an opinion that does not corresponds to an alternative, as it is the case, whenever opinions are expressed, e.g., through polls. Observe that $\vert \Theta \vert = 2 \left\lceil\frac{1}{\delta}\right\rceil + 1$. Let $\bfx = (x_1, \ldots, x_n)$ be a profile of opinions held by players, where $x_i$ is the opinion kept of player $i$. 

The opinions of agents are influenced by their social relationships. Specifically, we assume that, for each edge $(i,j) \in E$, opinions of agents $i$ and $j$ are mutually influenced and the weight $w_{ij} > 0$ of the edge models the strength of this influence.

Moreover, we assume that the opinion of an agent can be also influenced by recommendations received directly from the social media and not diffused through their own neighbours. We assume that the social media can present to the agents different recommendations, tailoring them on their interests. In particular, we assume that the social media has a discrete subset $\Omega$ of $[-1,1]$, representing the available information, and it decides to present to an agent with opinion $x$ the information $s(x) \in \Omega$, where the function $s \colon \Theta \rightarrow \Omega$ models the recommendation procedure adopted by the media. Clearly, since the social media is interested in increasing the engagement of their users to the platform, it is interested in advertising to users information that best matches their profile.
Thus, e.g., in an electoral setting, the social media will recommend right parties to right-oriented agents, left parties to to left-oriented agents, and moderate party to remaining agents.
% Let's back to our example. Consider an agent with opinion $x$, $s(x)$ represent the content presented to the agent by the social media. If $s(x) = 1$, then the social media present to the agent a content that support the extreme right party

Thus, at each time step $t$ agents update their opinions depending on the opinions held by their social relations and the recommendations received by the social media. We denote by $\bfx^t$ the profile of opinions held by agents at time $t$.

In this work, following the model introduced by \cite{immorlica}, we will consider a specific choice for $\Omega$ and $s$: in particular, we assume $\Omega = \{-1, 0, 1\}$ (we will sometimes refer to the elements of $\Omega$ as ``extreme left'', ``extreme moderate'', and ``extreme right'' information or opinions), and assume $s$ being a symmetric threshold function such that $s(x) = -1$ if $x < -\lambda$, $s(x) = 1$ if $x > \lambda$, and $s(x) = 0$ otherwise, for some $0 < \lambda < 1$. While this choice is clearly simplifying the model, it still leads to interesting results about how these social media recommendations may affect the chance that agents may reach a consensus.
Moreover, w.l.o.g., we will assume that $\lambda = 1/2$. This essentially means that the social media shows to each agent the information that is closest to her opinion (by breaking ties in favour of the ``moderate'' information). We remark that all our results about the impact of social media recommendations may be easily extended to arbitrary values of $\lambda$.

The combined influence of neighbours and social media recommendations may lead an agent to update her opinion. In this work, we follow the principles of the model presented by \cite{DeGroot} to represent how the opinion is updated. Specifically, since our focus is on a setting with discrete opinions, we will adapt to our model the discrete generalization of the DeGroot model defined by \cite{goldberg}: at each step $t \geq 1$, agent $i$ will choose the opinion $x$ that minimizes
% \begin{equation}
% \label{eq:utility}
$c_i(x,\bfx^{t-1}) = b(x - s(x_i^{t-1}))^2 + \sum_{j \colon (i,j) \in E} w_{ij} (x - x_j^{t-1})^2$,
% \end{equation}
where $b > 0$ is the weight of the influence of the social media on agents, and $\bfx^{t-1} = (x^{t-1}_1, \ldots, x^{t-1}_n)$ is the opinion profile at the previous time step. We notice that this setting can be equivalently described as a game: agents are the players, opinions are their strategies, and the function $c_i$ is the cost function of player $i$. According to this game-theoretic viewpoint, the opinion update consists essentially of selecting the \emph{best-response} strategy, i.e. the one that minimizes the cost of the player given the strategies currently selected by other players and the social media.

We say that an opinion profile $\bfx^t = (x_1^t, \ldots, x_n^t)$ is a \emph{consensus} (on opinion $\overline{x}$)  if $x_i^t = \overline{x}$ for every $i$. Moreover, we say that an opinion profile $\bfx^t = (x_1^t, \ldots, x_n^t)$ is \emph{stable} if it is a Nash equilibrium of the corresponding game, i.e. $x_i^t$ minimizes $c_i(x,x^{t})$ for every agent $i$. It is easy to see that a consensus on an extreme opinion, say, e.g., $1$, is always a stable profile. Hence, in this opinion game, a Nash Equilibrium always exits.

Although a stable profile always exists, for given $G$ and $b$, there may be multiple stable opinion profiles, and which one is reached depends on the way in which agents update their opinions.
In the literature, the DeGroot model has been associated to different update rules:
% , depending on how it is selected the set of agents that at time $t$ is allowed to change their opinions.
The most popular rules are: i) \emph{synchronous rule}, where at each time step $t$ all the agents update their options; ii) \emph{asynchronous rule}, where at each time step $t$, a single agent, arbitrarily chosen, is allowed to update his/her opinion. %We will show that only in \emph{asynchronous} case the evolution of opinions always leads a stable opinion profile.\\

In the next section we will focus on the \emph{synchronous} case.
We will analyze the dynamics with the \emph{asynchronous} update rule later in section \ref{seAsync}.
% In this paper we will focus on \emph{synchronous} updates.
% the two extremal order of updates: in the \emph{synchronous} case, ; in the \emph{asynchronous} case, at each time step only one agent is allowed to update their opinion. While, as we will show below, our findings about the effect of recommendations on opinion consensus in the synchronous case appear to extend even to the asynchronous case, we will show that only the latter the evolution of opinions always leads a stable opinion profile.

%CITARE LE DUE VERSIONI

\section{Synchronous Updates}
\label{sync}
%   In this section we will focus on the case in which all agents update their opinion at each time step.
In this section we will analyze the dynamics when updates are synchronous. As we will see, the synchronism of the updates allows representing the opinions' dynamics in a simple and tractable way. 
Through the analysis of the dynamics with synchronous updates, we obtain interesting findings about the effect of social media recommendations on opinion consensus. 
We will see in section \ref{seAsync} that these findings extend even to the asynchronous case. 
\subsection{Symmetric Two-Block Model}
\label{sec3}
We will start our study by focusing on a simple setting: an \emph{$(\ain,\aout)$-symmetric two-block model}. This is defined as follows: given an undirected graph $G=(V, E)$ and a value $b \geq 0$, we partition the set $V$ of agents in two subsets, $L$ and $R$ such that, for each agent $u \in P$ with $P \in \{L, R\}$, we set $x_u^0 = x_P^0$.
Moreover, we set weights $w_{ij}$ for each edge $(i,j) \in E$, such that for each agent $i \in P$, we have that $\sum_{\begin{subarray}{c}j \in P\\(i,j) \in E\end{subarray}} w_{ij} = \ain$ and $\sum_{\begin{subarray}{c}j \in \notP\\(i,j) \in E\end{subarray}} w_{ij} = \aout$, where $\notP = \{L, R\} \setminus P$, and $\ain, \aout > 0$. Roughly speaking, in a symmetric two-block model we assume that agents hold only two opinions and we can partitionate them in two symmetric communities depending on their opinions. Moreover, the cumulative influence that an agent receives from members of her own community is the same for each agent, namely $\ain$. Similarly, the influence that an agent receives from members of the opposite community is the same for each agent, namely $\aout$. The ratio $h = \frac{\ain}{\aout}$ is sometimes termed \emph{homophily ratio} \citep{dandekar}, and the ratio $\tb = \frac{b}{\aout}$ is termed the \emph{relative amount of media influence}. These quantities will play an important role in our analysis.
% in understanding how social media may affect the possibility of reaching consensus.

We will investigate on how the influence $b$ of the social network and the homophily ratio $h$ affect the probability of reaching a consensus in this setting and the type of the consensus obtained.
In particular, we will show that if either the homophily ratio $h$ or the media influence $b$ are very large, consensus is very hard to achieve when the initial opinions in the communities are divergent,
% \footnote{As we will show below, a consensus on opinion $0$ may be possible for some values of $b$.},
i.e., $x_i^0 < 0$ for each $i \in P$, and $x_j^0 > 0$ for each $j \in \notP$. This result follows from the fact that both homophily and media influence tend to extremize the opinions of the two groups, by leading them to diverge.

Moreover, we will show that, if the homophily ratio is large, then convergence to consensus becomes hard even for low values of media influence. Interestingly, this latter holds regardless of the number of opinions in $\Theta$: in particular, it holds even if this number is very large (and thus the parameter $\delta$ is very small). However, for small values of $\delta$ and $b$ our model resembles the DeGroot model, for which it is known that a consensus is always reached. In other words, our results prove that the consensus property of DeGroot model is not robust even to a small discretization of the opinion space.

Moreover, we also study the type of the consensus reached by the players.
In particular, we show that when $b$ is great, consensus is only possible on extreme opinions, namely $-1$, $0$ and $1$, even for small values of $h$, and even for non-diverging initial opinion profiles.

In conclusion, our results show that in the two-block model the effect of the social media influence is limited on communities with a large homophily ratio. However, when agents become more prone to heterogeneous influence, then the social media may play an important role, by making consensus either harder to reach, or reachable only on extreme opinions.

% We conjecture that our results can be generalized to more realistic settings and in next section we give some experimental evidence to support our conjecture.
\subsubsection{Preliminary Results}
\label{preliminary}
In the following, we will give some useful characterizations of \emph{feasible} opinion profiles, i.e., profiles that can be reached during the evolution of the dynamics, and stable profiles.
\begin{lemma}
\label{lem:TwoPlayerModel}
Given an $(\ain,\aout)$-symmetric two-block model $G=(L \cup R, E, w)$ and a social media influence $b$, an opinion profile $(x_1, \ldots, x_n)$ is feasible only if $x_i = x_L$ for every $i \in L$, $x_j = x_R$ for every $j \in R$.
\end{lemma}
\begin{proof}
% POSSIAMO ASSUMERE LAMBDA BELONGS TO OMEGA\\
From the definition of cost,
% \eqref{eq:utility},
the opinion of agent $i$ at time $t+1$ is equal to 
\begin{equation*}
x_i^{t+1} = \argmin_{y \in \Theta}\left\{b(y - s(x_i^t))^2 + \sum_{j: (i,j) \in E} w_{i,j} (y-x_j^t)^2\right\}.
\end{equation*}
Therefore, in an $(\ain,\aout)$-symmetric two-block model, for each agent $i \in P$ the opinion at step $1$ is: 
\begin{equation*}
x_i^{1} = \argmin_{y \in \Theta}\{b(y - s(x_P^0))^2 + \ain(y-x_P^0)^2 + \aout (y-x_{\notP}^0)^2\}.
\end{equation*}
Consequently, at step $1$ $x_i^1 = x_L^1$ for every $i \in L$, and $x_j^1=x_R^1$ for every $j \in R$.

The lemma follows by iteratively applying the same argument for all the following steps.
\end{proof}
% \textcolor{red}{Essentially, by Lemma~\ref{lem:TwoPlayerModel} we can say that an $(\ain,\aout)-$symmetric two block-model its equivalent to a simple $(L,R)$-two player model}.\\
Next, we provide a characterization of best responses. This is our key lemma.
\begin{lemma}
\label{lem:boundOnAvg}
 Given an $(\ain,\aout)$-symmetric two-block model $G=(L \cup R, E, w)$, a social media influence $b$, and a feasible opinion profile $\bfx$, $x^*$ is a best-response for agent $i \in P$, for $P \in \{L, R\}$, in the profile $\bfx$ only if
 \begin{equation}
  x^* - \frac{\delta}{2} \leq \frac{bs(x_P)+\aout x_{\notP}+\ain x_P}{b+\aout+\ain} \leq x^* + \frac{\delta}{2},
 \end{equation}
 where $\notP = \{L, R\} \setminus P$, $x_L = x_i$ for some $i \in L$, and $x_R = x_j$  for some $j \in R$.
\end{lemma}
\begin{proof}
By Lemma \ref{lem:TwoPlayerModel}, all the agents in the same community $P$ have the same best response in $\bfx$, and, from the definition of cost, 
% \eqref{eq:utility},
it is
\begin{align*}
x^* & = \argmin_{y \in \Theta}\left\{b(y - s(x_P^t))^2 + \ain(y-x_P^t)^2 + \aout (y-x_{\notP}^t)^2\right\}\\
& = \argmin_{y \in \Theta}\left\{(\aout+\ain+b)y^2-2y(bs(x_P^t)+\aout x_{\notP}^t+\ain x_P^t) + \right.\\
& \qquad \qquad \qquad \left.bs(x_P^t)^2+\aout(x_{\notP}^t)^2+\ain(x_P^t)^2\right\}.
\end{align*}
Since, by definition, $\aout+\ain+b > 0$, then the argument of $\argmin$ describes a parabola with concavity upwards. If $\Theta$ was a continuous interval, the minimum value would be achieved by setting $y = \frac{bs(x_P^t)+\aout x_{\notP}^t+\ain x_P^t}{\aout+\ain+b} \in [-1, 1]$, where the membership in this interval follows since $\vert s(x_P^t) \vert \leq 1$, $\vert x_{\notP}^t \vert \leq 1$, and $\vert x_P^t \vert \leq 1$. However, since $\Theta$ is discrete, then we have that
\begin{equation*}
x^{*} = \argmin_{y \in \Theta}\left\{\left(y - \frac{bs(x_P)+\aout x_{\notP}+\ain x_P}{b+\aout+\ain}\right)^2\right\}.
\end{equation*}
Since $\delta$ is the distance among two consecutive elements in $\Theta$, then there is an $y \in \Theta$ such that $\left\vert y - \frac{bs(x_P)+\aout x_{\notP}+\ain x_P}{b+\aout+\ain}\right\vert \leq \frac{\delta}{2}$, from which the lemma follows.
\end{proof}
Finally, next lemma provides a characterization of the stable opinion profiles. To this aim, we define the \emph{relative amount of media influence} as $\tb = \frac{b}{\aout}$.
\begin{lemma}
\label{lem:base_systems}
 Given an $(\ain,\aout)$-symmetric two-block model $G=(L \cup R, E, w)$, and a social media influence $b$, a feasible opinion profile $\bfx$ is stable only if the following conditions are satisfied:
 \begin{equation*}
  \begin{cases}
   \tb \left(s(x_L) - x_L - \frac{\delta}{2}\right) \leq \left(x_L - x_R\right) + \frac{\delta}{2} \left(h + 1\right);\\
   \tb \left(s(x_L) - x_L + \frac{\delta}{2}\right) \geq \left(x_L - x_R\right) - \frac{\delta}{2} \left(h + 1\right);\\
   \tb \left(s(x_R) - x_R - \frac{\delta}{2}\right) \leq \left(x_R - x_L\right) + \frac{\delta}{2} \left(h + 1\right);\\
   \tb \left(s(x_R) - x_R + \frac{\delta}{2}\right) \geq \left(x_R - x_L\right) - \frac{\delta}{2} \left(h + 1\right).
  \end{cases}
 \end{equation*}
\end{lemma}
\begin{proof}
By Lemma \ref{lem:TwoPlayerModel} an opinion profile is feasible only if every agent $i \in P$ with $P \in \{L, R\}$ has the same opinion $x_P$. Moreover, this opinion profile is stable if for every agent $i \in P$ its best response with respect to $\bfx$ is still $x_P$. By Lemma~\ref{lem:boundOnAvg} we then have that if $\bfx$ is stable then it must hold that:
\begin{equation*}
\begin{cases}
x_L - \frac{\delta}{2} \leq \frac{bs(x_L)+\aout x_{R}+\ain x_L}{b+\aout+\ain} \leq x_L + \frac{\delta}{2};\\
x_R - \frac{\delta}{2} \leq \frac{bs(x_R)+\aout x_{L}+\ain x_R}{b+\aout+\ain} \leq x_R + \frac{\delta}{2}.
\end{cases}
\end{equation*}
By simple algebraic manipulations, we can rewrite above conditions as follows: 
\begin{equation*}
  \begin{cases}
   \frac{b}{\aout} \left(s(x_L) - x_L - \frac{\delta}{2}\right) \leq \left(x_L - x_R\right) + \frac{\delta}{2} \left(\frac{\ain}{\aout} + 1\right);\\
   \frac{b}{\aout} \left(s(x_L) - x_L + \frac{\delta}{2}\right) \geq \left(x_L - x_R\right) - \frac{\delta}{2} \left(\frac{\ain}{\aout} + 1\right);\\
   \frac{b}{\aout} \left(s(x_R) - x_R - \frac{\delta}{2}\right) \leq \left(x_R - x_L\right) + \frac{\delta}{2} \left(\frac{\ain}{\aout} + 1\right);\\
   \frac{b}{\aout} \left(s(x_R) - x_R + \frac{\delta}{2}\right) \geq \left(x_R - x_L\right) - \frac{\delta}{2} \left(\frac{\ain}{\aout} + 1\right).
  \end{cases}
 \end{equation*}
The lemma follows by observing that $\tb = \frac{b}{\aout}$ and $\frac{\ain}{\aout} = h$.
\end{proof}

Next lemmas prove some properties of the best-response opinion in presence of a social media influence that will be useful in characterizing the type of stable profile achieved by our opinion dynamics. 
Specifically, Lemma~\ref{lem:far} provides conditions for an agent with an opinion close to the extremes, $-1$ and $1$, to change her idea and adopt an opinion that is far from these extremes. Lemma~\ref{lem:close}, Lemma~\ref{lem:close_zero}, Lemma~\ref{lem:zero}, and Lemma~\ref{lem:far_zero} provide instead conditions for an agent to adopt an opinion of opposite sign with respect to her actual opinion. 
All these lemmas are proved in a very similar way, hence, for the seek of readability, we will only show proofs of Lemma~\ref{lem:far} and Lemma~\ref{lem:close}. We refer interested reader to Appendix~\ref{apx:pre_res} for the remaining proofs. 
In what follows we will assume that $\lambda \in \Theta$. This is without loss of generality: indeed, all results below still hold by simply replacing $\lambda$ with the largest $\lambda' \in \Theta$ smaller than $\lambda$, since for every $x \in \Theta$ it must be the case that $x > \lambda$ (resp. $x < -\lambda$) if and only if $x > \lambda'$ (resp. $x < -\lambda'$).
\begin{lemma}
 \label{lem:far}
 Given an $(\ain,\aout)$-symmetric two-block model $G=(L \cup R, E, w)$ and a social media influence $b$, if $\left\vert x_i^t \right\vert > \lambda$, then the opinion of the agent $i$ at next step is $\left\vert x_i^{t+1}\right\vert \leq \lambda$ only if $\tb \leq \tau_1(h)$, where $\tau_1(h) = \frac{2+2\lambda + \delta - \delta h}{2-2\lambda-\delta}$.
\end{lemma}
\begin{proof}
We consider only the case that $x_i^t > \lambda$. The case for $x_i^t < -\lambda$ is symmetric and hence omitted.

Recall that, by Lemma~\ref{lem:TwoPlayerModel}, $x_i^t = x_P^t$, where $P$ is the block $i$ belongs to, and every $j \notin P$ has $x_j^t = x_{\notP}^t$.
Since $x_i^t > \lambda$, then $s(x_P^t) = 1$. Thus, by Lemma \ref{lem:boundOnAvg}, $x_i^{t+1} \leq \lambda$ only if $\frac{b+\aout x_{\notP}^t+\ain x_P^t}{b+\ain+\aout} \leq \lambda+\frac{\delta}{2}$.
By dividing both sides by $\aout$ and recalling that $\frac{b}{\aout} = \tb$ and $\frac{\ain}{\aout} = h$, we have that $x_i^{t+1} \leq \lambda$ only if
\begin{equation}
\label{eq:condlem4}
 \tb \leq \frac{-x_{\notP}^t+\lambda+\frac{\delta}{2}+h(-x_P^t+\lambda+\frac{\delta}{2})}{1-\frac{\delta}{2}-\lambda}.
\end{equation}
It is immediate to check that the r.h.s. of \eqref{eq:condlem4} is maximized by taking $x_P^t = \lambda + \delta$, and $x_{\notP}^t = -1$. By substituting these values in \eqref{eq:condlem4}, we achieve that 
$x_i^{t+1}\leq \lambda$ only if $\tb \leq \frac{2+2\lambda + \delta - \delta h}{2-2\lambda-\delta}$, as desired.
% \item If $x_i^t < -\lambda$, then $s(x_i^t) = -1$. By Lemma \ref{lem:boundOnAvg}, the next step opinion $x_i^{t+1}$ will be higher than or equal to $-\lambda$ only if:
% \begin{equation*}
% \frac{-b+\aout x_{\notP}^t+\ain x_P^t}{b+\ain+\aout} \geq -\lambda- 		    \frac{\delta}{2}
% \end{equation*}
% Grouping by $\frac{b}{\aout} = \tb$ we have the following inequality:
% \begin{equation*}
% \tb \leq \frac{x_{\notP}^t+\lambda+\frac{\delta}{2}+h(x_P^t+\lambda+\frac{\delta}{2})}{1-\frac{\delta}{2}-\lambda}
% \end{equation*}
% Consequently, if $\tb > \max_{(x_P^t, x_{\notP}^t)}{\frac{x_{\notP}^t+\lambda+\frac{\delta}{2}+h(x_P^t+\lambda+\frac{\delta}{2})}{1-\frac{\delta}{2}-\lambda}}=_{(x_P^t = -\lambda-\delta, x_{\notP} = 1)} \frac{2+2\lambda+\delta-\delta h}{2-2\lambda-\delta}$, then the next step opinion can not be higher than or equal to $-\lambda$.
\end{proof}

\begin{lemma}
\label{lem:close}
 Given an $(\ain,\aout)$-symmetric two-block model $G=(L \cup R, E, w)$ and a social media influence $b$, if $x_i^t \in [-\lambda,0]$ (resp., $x_i^t \in [0,\lambda]$), then $x_i^{t+1} > 0$ (resp., $x_i^{t+1} < 0$) only if $\tb \leq \tau_2(h)$, where $\tau_2(h) = \left(\frac{2}{\delta} - 1\right) - h$.
\end{lemma}
\begin{proof}
We consider only the case that $x_i^t \in [-\lambda,0]$. The case for $x_i^t \in [0,\lambda]$ is symmetric and hence omitted.

Recall that, by Lemma~\ref{lem:TwoPlayerModel}, $x_i^t = x_P^t$, where $P$ is the block at which $i$ belongs, and every $j \notin P$ has $x_j^t = x_{\notP}^t$.
Since $x_i^t \in [-\lambda,0]$, then $s(x_P^t) = 0$. Thus, by Lemma \ref{lem:boundOnAvg}, $x_i^{t+1} \geq 0$ only if $\frac{\aout x_{\notP}^t+\ain x_P^t}{b+\aout+\ain} \leq \frac{\delta}{2}$.
By dividing both sides by $\aout$ and recalling that $\frac{b}{\aout} = \tb$ and $\frac{\ain}{\aout} = h$, we have that $x_i^{t+1}\geq 0$ only if 
\begin{equation}
\label{eq:condlem5}
 \tb \leq \frac{2}{\delta}\left(x_{\notP}^t-\frac{\delta}{2}+h(x_P^t-\frac{\delta}{2})\right).
\end{equation}
It is immediate to check that the r.h.s. of \eqref{eq:condlem5} is maximized by taking $x_P^t = 0$, and $x_{\notP}^t = 1$. By substituting these values in \eqref{eq:condlem5}, we achieve that 
$x_i^{t+1}\geq 0$ only if $\tb \leq \left(\frac{2}{\delta} - 1\right) - h$, as desired.
% \item if $x_i^t \in [0, \lambda]$, then $s(x_i^t) = 0$. By Lemma \ref{lem:boundOnAvg} the next step opinion $x_i^{t+1}$ will be lower than zero only if:
% \begin{equation*}
% \frac{\aout x_{\notP}^t+\ain x_P^t}{b+\aout+\ain} \leq - \frac{\delta}{2}
% \end{equation*}
% Grouping by $\frac{b}{\aout} = \tb$, we have the following inequality:
% \begin{equation*}
% \tb \leq -\frac{2}{\delta}\left(x_{\notP}^t+\frac{\delta}{2}+h(x_P^t+\frac{\delta}{2})\right)
% \end{equation*}
% Therefore, if $\tb > \max_{( x_P^t, x_{\notP}^t)}{\left( -\frac{2}{\delta}\left(x_{\notP}^t+\frac{\delta}{2}+h(x_P^t+\frac{\delta}{2})\right)\right)} =_{(x_P^t = 0, x_{\notP} = -1)} \frac{2}{\delta}-1-h$, then the next step opinion can not be lower than zero.
\end{proof}
\begin{lemma}
\label{lem:close_zero}
 Given an $(\ain,\aout)$-symmetric two-block model $G=(L \cup R, E, w)$ and a social media influence $b$, if $x_i^t < -\lambda$ (resp., $x_i^t > \lambda$), then $x_i^{t+1} > 0$ (resp., $x_i^{t+1} < 0$) only if $\tb \leq \tau_3(h)$, where $\tau_3(h) = \frac{2-\delta-(2\lambda+3\delta)h}{2+\delta}$.
\end{lemma}

\begin{lemma}
\label{lem:zero}
 Given an $(\ain,\aout)$-symmetric two-block model $G=(L \cup R, E, w)$ and a social media influence $b$, if $x_i^t \in (0, \lambda]$ (resp., $x_i^t \in [-\lambda,0)$), then $x_i^{t+1} \cdot x_i^t \leq 0$ only if $\tb \geq \tau_4(h)$, where $\tau_4(h) = h - \left(\frac{2}{\delta} + 1\right)$.
\end{lemma}

\begin{lemma}
\label{lem:far_zero}
 Given an $(\ain,\aout)$-symmetric two-block model $G=(L \cup R, E, w)$ and a social media influence $b$, if $\left\vert x_i^t \right\vert > \lambda$, then $x_i^{t+1} \cdot x_i^t \leq 0$ only if $\tb \leq \tau_5(h)$, where $\tau_5(h) = \frac{2+\delta-(2\lambda+\delta)h}{2-\delta}$.
\end{lemma}

We next show that some interesting relationships exist among these thresholds.
\begin{lemma}
\label{lem:relation}
 The following relationships hold:
 \begin{align}
  \tau_1(h) & > \tau_4(h)& \text{only if $h < \frac{2}{\delta} + \frac{1}{1-\lambda}$}; \label{eq:t1t3}\\
  \tau_2(h) & < \tau_4(h)& \text{if $h > \frac{2}{\delta}$};\label{eq:t2t3}\\
  \tau_3(h) & < \tau_5(h)& \text{for every $h$};\label{eq:t3t5}\\
%   \tau_2(h) & > 0& \text{only if $h < \frac{2}{\delta}-1$};\\%\label{eq:t20}\\
%   \tau_3(h) & > 0& \text{only if $h < \frac{2-\delta}{2\lambda+3\delta}$};\label{eq:t30}\\
%   \tau_4(h) & > 0& \text{only if $h > \frac{2}{\delta}+1$};\label{eq:t40}\\
  \tau_5(h) & < 0& \text{if $\tau_4(h) > 0$}.\label{eq:t3t4}
 \end{align}
\end{lemma}
\begin{proof}
\begin{description}
 \item[\eqref{eq:t1t3}:] This relationship immediately follows by observing that
 \begin{align*}
  \tau_1(h) - \tau_4(h) & = \frac{2+2\lambda + \delta - \delta h}{2-2\lambda-\delta} - h + \left(\frac{2}{\delta} + 1\right)\\
  & = \frac{2}{2-2\lambda-\delta} \left[\frac{2}{\delta} (1-\lambda) + 1 - h(1-\lambda)\right].
 \end{align*}
 \item[\eqref{eq:t2t3}:] This relationship immediately follows by observing that
 \begin{equation*}
  \tau_2(h) - \tau_4(h) = \left(\frac{2}{\delta} - 1\right) - h - h + \left(\frac{2}{\delta} + 1\right) = 2 \left(\frac{2}{\delta} - h\right).
 \end{equation*}
 \item[\eqref{eq:t3t5}:] This relationship immediately follows by observing that
 \begin{align*}
  \tau_3(h) - \tau_5(h) & = \frac{2-\delta-(2\lambda+3\delta)h}{2+\delta} - \frac{2+\delta-(2\lambda+\delta)h}{2-\delta}\\
  & = \frac{4\delta}{4-\delta^2} \left[(\lambda+\delta-1)h - 1\right] < 0,
 \end{align*}
 where the last inequality follows because $\lambda + \delta < 1$.
 \item[\eqref{eq:t3t4}:] If $\tau_4(h) > 0$, the $h > \frac{2}{\delta}+1$.
 Then
 $\tau_5(h) = \frac{2+\delta-(2\lambda+\delta)h}{2-\delta} < - \frac{2\lambda(2+\delta)}{\delta(2-\delta)} < 0$.\qedhere
\end{description}
\end{proof}

% \textcolor{red}{
% By Lemma \ref{lem:zero} and Lemma \ref{lem:far_zero} we can simple derive the following corollary.
% \begin{corollary}
% Given an $(\ain,\aout)$-symmetric two-block model $G=(L \cup R, E, w)$ and a social media influence $b$, if agent $i$ at time $t$ has opinion $x_i^t \neq 0$, if $\tau_5(h) \leq \tb \leq \tau_4(h)$, then $x_i^{t+1} \cdot x_i^t \geq 0$.
% \end{corollary}
% \begin{proof}
% The corollary is a trivial combination of  \ref{lem:zero} and Lemma \ref{lem:far_zero}.
% \end{proof}
% \begin{remark}
% The interval $\tau_5(h) \leq \tb \leq \tau_4(h)$ is not empty only if $\tau_4(h) > 0$ 
% \end{remark}
% }

\subsubsection{Consensus Characterization}
\label{sync_results}
% In Appendix~\ref{apx:preliminary} we will discuss some preliminary results aiming to describe for which values of $h$ and $b$ a consensus can be reached.
In this subsection, we will
% deepen our analysis by studying
study which type of consensus can be achieved in a two-block model, depending on the the homophily ratio and the social media influence. We will distinguish different cases, depending on the initial opinions held by the players in the two blocks. To this aim let us define the following quantities that will play a fundamental role in our characterization: $\tau_1(h) = \frac{2+2\lambda + \delta - \delta h}{2-2\lambda-\delta}$, $\tau_2(h) = \left(\frac{2}{\delta} - 1\right) - h$, $\tau_3(h) = \frac{2-\delta-(2\lambda+3\delta)h}{2+\delta}$, $\tau_4(h) = h - \left(\frac{2}{\delta} + 1\right)$, and $\tau_5(h) = \frac{2+\delta-(2\lambda+\delta)h}{2-\delta}$.

\paragraph*{Divergent and Extreme Initial Opinions}
% \smallskip \noindent \emph{Divergent and Extreme Initial Opinions.}
We start by considering the case where the starting opinions of the two blocks diverge (i.e., one is positive and the other is negative) and are both far away from $0$.
\begin{theorem}
 \label{thm:div_large_both}
 Given an $(\ain,\aout)$-symmetric two-block model $G=(L \cup R, E, w)$ and a social media influence $b$, if $\left\vert  x_L^0\right\vert  > \lambda$ and $\left\vert  x_R^0\right\vert  > \lambda$ and $x_L^0 \cdot x_R^0 < 0$, then
  {\small
 \begin{equation}
 \label{eq:cases}
  \begin{cases}
   \text{if $\tb > \tau_1(h)$,} & \text{no consensus can be stable};\\
   \text{if $\max\left\{0, \tau_2(h), \tau_3(h), \tau_4(h)\right\} < \tb \leq \tau_1(h)$,} & \text{only consensus on $0$ can be stable};\\
   \text{if $\max\left\{0, \tau_4(h)\right\} < \tb \leq \max\{\tau_2(h),\tau_3(h)\}$,} & \text{non-extreme consensus can be stable};\\
   \text{if $0 < \tb \leq \max\left\{0, \tau_4(h)\right\}$,} & \text{no consensus can be stable}.
  \end{cases}
 \end{equation}}

\begin{proof}
 Suppose first that $\tb > \tau_1(h)$. Then, by Lemma~\ref{lem:far}, no agent $i \in L$ can take an opinion $x_L^1$ such that $\left\vert  x_L^1\right\vert  \leq \lambda$, and no agent $j \in R$ can take an opinion $x_R^1$ such that $\left\vert  x_L^1 \right\vert  \leq \lambda$.
 Hence, after the first time step $\left\vert x_L^1\right\vert > \lambda$ and $\left\vert x_R^1\right\vert > \lambda$ and $x_L^1 \cdot x_R^1 < 0$. Then, we can iteratively apply the same argument to conclude that the opinions of the two blocks never converge to a consensus profile.
 
 Suppose now that $\max\left\{0, \tau_2(h), \tau_3(h), \tau_4(h)\right\} < \tb \leq \tau_1(h)$. W.l.o.g., suppose that $x_L^0 \leq 0$ and $x_R^0 \geq 0$. Since $\tb > \max\{\tau_2(h), \tau_3(h)\}$, then, by Lemma~\ref{lem:close} and Lemma~\ref{lem:close_zero}, it follows that no agent $i \in L$ can take an opinion $x_L^1 > 0$, and no agent $j \in R$ can take an opinion $x_R^1 < 0$. Hence, after the first time step $x_L^1 \leq 0$ and $x_R^1 \geq 0$. Then, we can iteratively apply the same argument above to conclude that the unique opinion on which the two blocks can converge is $0$.
 
 Finally, suppose that $0 < \tb \leq \max\left\{0, \tau_4(h)\right\}$. Note that this interval is non-empty only if $\tau_4(h) > 0$, and thus, according to \eqref{eq:t3t4}, $\tau_5(h) < 0$.
 W.l.o.g., suppose that $x_L^0 < 0$ and $x_R^0 > 0$. Then, by Lemma~\ref{lem:zero} and Lemma~\ref{lem:far_zero}, it follows that no agent $i \in L$ can take an opinion $x_L^1 \geq 0$, and no agent $j \in R$ can take an opinion $x_R^1 \leq 0$. Hence, after the first time step $x_L^1 < 0$ and $x_R^1 > 0$. Then, we can iteratively apply the same argument above to conclude that the two blocks never converge to a consensus.  
\end{proof}
\end{theorem}

%sketch proof
\begin{comment}
\begin{proof}[Sketch]
 Suppose first that $\tb > \tau_1(h)$. Then, it may be showed that no agent $i \in L$ can take an opinion $x_L^1$ such that $\left\vert x_L^1\right\vert \leq \lambda$, and no agent $j \in R$ can take an opinion $x_R^1$ such that $\left\vert x_L^1\right\vert \leq \lambda$.
 Hence, after the first time step $\left\vert x_L^1\right\vert > \lambda$ and $\left\vert x_R^1\right\vert > \lambda$ and $x_L^1 \cdot x_R^1 < 0$. Then, we can iteratively apply the same argument to conclude that the opinions of the two blocks never converge to a consensus profile.
 
 Suppose now that $\max\left\{0, \tau_2(h), \tau_3(h), \tau_4(h)\right\} < \tb \leq \tau_1(h)$. W.l.o.g., suppose that $x_L^0 \leq 0$ and $x_R^0 \geq 0$. Since $\tb > \max\{\tau_2(h), \tau_3(h)\}$, then, it may be showed that no agent $i \in L$ can take an opinion $x_L^1 > 0$, and no agent $j \in R$ can take an opinion $x_R^1 < 0$. Hence, after the first time step $x_L^1 \leq 0$ and $x_R^1 \geq 0$. Then, we can iteratively apply the same argument above to conclude that the unique opinion on which the two blocks can converge is $0$.
 
 Finally, suppose that $0 < \tb \leq \max\left\{0, \tau_4(h)\right\}$. As above, it may be showed that no agent $i \in L$ can take an opinion $x_L^1 \geq 0$, and no agent $j \in R$ can take an opinion $x_R^1 \leq 0$. Hence, after the first time step $x_L^1 < 0$ and $x_R^1 > 0$. Then, we can iteratively apply the same argument above to conclude that the two blocks never converge to a consensus. 
 
\end{proof}
\end{comment}

\begin{remark}
\label{rem:beh_k}
 We observe that it is impossible that the interval corresponding to last two cases of \eqref{eq:cases} are both non-empty. Indeed, if the last interval is non-empty, then $\tau_4(h) > 0$ and thus $h > \frac{2}{\delta} + 1$.
%  This implies, by \eqref{eq:t2t3}, \eqref{eq:t3t4} and \eqref{eq:t3t5}, 
 It is not hard to check that this implies
 that $\max\{\tau_2(h), \tau_3(h)\} < \tau_4(h) = \max\{0, \tau_4(h)\}$. Hence for $\tb \leq \max \{0, \tau_2(h), \tau_3(h), \tau_4(h)\}$, either no consensus can be stable, or it is possible to achieve consensus also on non-extreme opinions. 
\end{remark}
Roughly speaking, Theorem~\ref{thm:div_large_both} shows that if we have initial opinions that are divergent and far away from $0$, consensus is impossible to achieve for high values of the media influence $b$, while it can be achieved on non-extremal opinions only for small values of $b$ and under opportune conditions.
Remark~\ref{rem:beh_k} also shows that the outcome also depends on the value of $h$. Specifically, for small values of $h$, it appears that the chance of having a consensus decreases as $\tb$ increases, since we go from a range in which non-extreme consensus can be stable to a range in which only consensus on $0$ is stable, and finally to a range in which consensus is impossible. Instead, for large $h$ the behaviour is less ``monotone'': indeed, we go from no consensus to possible consensus (on opinion $0$) and again to no consensus. Moreover, from Theorem~\ref{thm:div_large_both} it's easy to prove that consensus is impossible when $h$ is large, regardless of the strength of the social media influence. This confirms the fundamental role played by the homophily ratio.

\paragraph*{Divergent Initial Opinion: Only One is Extreme}
% \smallskip \noindent \emph{Divergent Initial Opinion: Only One is Extreme.}
Consider now the case of divergent initial opinions, but we assume that only one of them is far from $0$. It will turn out that, as above, consensus is impossible whenever either the social media influence or the homophily ratio is large, and it is possible on non-extreme opinions only for small values of $b$ and under opportune conditions on $h$. 
The proof of Theorem~\ref{thm:div_large_one} is very similar to the proof of Theorem~\ref{thm:div_large_both}, hence, for the seek of readability, we postpone its proof to the Appendix~\ref{apx:sec3}.
\begin{theorem}
 \label{thm:div_large_one}
 Given an $(\ain,\aout)$-symmetric two-block model $G=(L \cup R, E, w)$ and a social media influence $b$, if $\left\vert x_L^0\right\vert > \lambda$ or $\left\vert x_R^0\right\vert > \lambda$ and $x_L^0 \cdot x_R^0 < 0$, then
  {\small
 \begin{equation*}
  \begin{cases}
   \text{if $\tb > \tau^*(h)$,} & \text{no consensus can be stable};\\
   \text{if $\max\left\{0, \tau_2(h), \tau_3(h), \tau_4(h)\right\} < \tb \leq \tau^*(h)$,} & \text{only consensus on $0$ can be stable};\\
   \text{if $\max\left\{0, \tau_4(h)\right\} < \tb \leq \max\left\{\tau_2(h), \tau_3(h)\right\}$,} & \text{non-extreme consensus can be stable};\\
   \text{if $0 < \tb \leq \max\left\{0, \tau_4(h)\right\}$,} & \text{no consensus can be stable},
  \end{cases}
 \end{equation*}}
 where $\tau^*(h) = \max\left\{\tau_1(h), \tau_2(h)\right\}$.
\end{theorem}
% The proof of this theorem is similar to the previous one, and hence omitted (see Appendix~\ref{apx:sec3} for details).
% \begin{proof}
%  Suppose w.l.o.g. that $x_L^0 < -\lambda $, and $x_R^0 \geq 0$.
%  Suppose first that $\tb > \tau^*(h)$. Since $\tau^*(h) \geq \tau_1(h)$, then, by Lemma~\ref{lem:far}, no agent $i \in L$ can take an opinion $x_L^1 \geq -\lambda$. Moreover, since $\tau^*(h) \geq \tau_2(h)$, then, by Lemma~\ref{lem:close}, no agent $j \in R$ can take an opinion $x_R^1 < 0$. Hence, after the first time step we still have $x_L^1 < -\lambda$ and $x_R^1 \geq 0$. Then, we can iteratively apply the same argument above to conclude that the opinions of the two blocks never converge to a consensus profile.
%  
%  As for the remaining cases, the argument is exactly the same as discussed in the proof of Theorem~\ref{thm:div_large_both}.
%  
% \end{proof}
The following corollary highlights that for large values of the homophily ratio convergence to consensus is impossible, regardless of the strength of the social media influence.
\begin{corollary}
\label{cor:h}
Given an $(\ain,\aout)$-symmetric two-block model $G=(L \cup R, E, w)$ and a social media influence $b$, if $\left\vert x_L^0\right\vert > \lambda$ or $\left\vert x_R^0\right\vert > \lambda$ and $x_L^0 \cdot x_R^0 < 0$, then no consensus opinion profile can be stable if $h \geq \frac{2}{\delta} + \frac{1}{1-\lambda}$, regardless of the value of the social media influence $b$. Moreover, consensus can be stable on non-extreme opinions only if $h < \max\left\{\frac{2}{\delta} - 1, \frac{2-\delta}{2\lambda + 3\delta}\right\}$.
\end{corollary}
\begin{proof}
If $h \geq \frac{2}{\delta} + \frac{1}{1-\lambda}$, we have that $\max\left\{\tau_2(h), \tau_3(h)\right\} \leq 0 \leq \max\{0, \tau_4(h)\}$, and, from \eqref{eq:t1t3}, $\tau_1(h) \leq \tau_4(h) \leq \max\{0, \tau_2(h), \tau_3(h), \tau_4(h)\}$. Hence, there cannot be any value of $\tb$ for which consensus can be stable.

Suppose instead that $h \geq \max\left\{\frac{2}{\delta} - 1, \frac{2-\delta}{2\lambda + 3\delta}\right\}$. Then, $\max\{\tau_2(h), \tau_3(h)\} \leq 0 \leq \max\{0, \tau_4(h)\}$, and thus the interval allowing for consensus on non-extreme opinions is empty.
\end{proof}

\paragraph*{Divergent Initial Opinion: Both are Moderate}
% \smallskip \noindent \emph{Divergent Initial Opinion: Both are Moderate.}
Consider now the case that initial opinions of the two blocks are still divergent but both close to opinion $0$. Clearly, in this case, large values of the social influence would push these opinions to $0$, thus leading to a consensus on this opinion. However, we show that this is the only possible consensus in this setting when $b$ is large. 
Proof of Theorem~\ref{thm:div_small} in a very similar to proof of Theorem~\ref{thm:div_large_both}; consequently, for the seek of readability, we postpone this proof to Appendix~\ref{apx:sec3}.
\begin{theorem}
 \label{thm:div_small}
 Given an $(\ain,\aout)$-symmetric two-block model $G=(L \cup R, E, w)$ and a social media influence $b$, if $\left\vert x_L^0\right\vert \leq \lambda$ and $\left\vert x_R^0\right\vert \leq \lambda$ and $x_L^0 \cdot x_R^0 < 0$, then
 {\small
 \begin{equation*}
  \begin{cases}
   \text{if $\tb > \max\left\{\tau_2(h), \tau_3(h)\right\}$,} & \text{only consensus on $0$ can be stable};\\
   \text{if $\max\left\{0, \tau_4(h)\right\} < \tb \leq \max\left\{\tau_2(h), \tau_3(h)\right\}$,} & \text{non-extreme consensus can be stable};\\
   \text{if $0 < \tb \leq \max\left\{0, \tau_4(h)\right\}$,} & \text{no consensus can be stable}.
  \end{cases}
 \end{equation*}}
\end{theorem}
%The proof follows exactly the same arguments as above, and hence it is omitted. 
% Similarly, as above, we can achieve the following corollary.
% \begin{corollary}
%  Given an $(\ain,\aout)$-symmetric two-block model $G=(L \cup R, E, w)$ and a social media influence $b$, if $\left\vert x_L^0\right\vert \leq \lambda$ and $\left\vert x_R^0\right\vert \leq \lambda$ and $x_L^0 \cdot x_R^0 < 0$, then no consensus opinion profile can be stable if $h \geq \frac{2}{\delta} + 1$, regardless of the value of the social media influence $b$.
% \end{corollary}

\paragraph*{Convergent Initial Opinions}
% \smallskip \noindent \emph{Convergent Initial Opinions.}
We conclude this section by considering the case that initial opinions do not diverge. We observe that, in this case, a large influence of the social media (with respect to homophily ratio) may lead only to consensus on extreme opinions, namely $-1, 0, 1$.
\begin{theorem}
 \label{thm:conv_small}
 Given an $(\ain,\aout)$-symmetric two-block model $G=(L \cup R, E, w)$ and a social media influence $b$, if $x_L^0 \cdot x_R^0 \geq 0$, then consensus on opinions different from $-1, 0, 1$ can be stable only if $\tb \leq h + 1$.
 
\begin{proof}
Suppose that consensus is achieved on $x \in \{\delta, \ldots, 1-\delta \}$ (the case in which consensus is achieved on a non-extreme negative opinion is symmetric and hence omitted).

Suppose first that $x \leq \lambda$. By Lemma \ref{lem:base_systems}, consensus on $x$ is a stable opinion profile only if both the following conditions hold:
\begin{equation*}
  \begin{cases}
   \tb \geq - \frac{\frac{\delta}{2}(1+h)}{x+\frac{\delta}{2}}\\
   \tb \leq \frac{\frac{\delta}{2}(1+h)}{x-\frac{\delta}{2}}
  \end{cases}
\end{equation*}
It is immediate to check that the first condition always holds, since $\tb, h, \delta, x \geq 0$. As for the second condition, its r.h.s. is maximized when $x = \delta$. Hence, a non-extreme consensus can be stable only if $\tb \leq \frac{\frac{\delta}{2}(1+h)}{\delta-\frac{\delta}{2}} = h + 1$, as desired.

Suppose now that $x > \lambda$. By Lemma \ref{lem:base_systems}, consensus on $x$ is a stable opinion profile only if both the following conditions hold:
\begin{equation*}
  \begin{cases}
   \tb \leq \frac{\frac{\delta}{2}(1+h)}{1-x-\frac{\delta}{2}}\\
   \tb \geq -\frac{\frac{\delta}{2}(1+h)}{1-x+\frac{\delta}{2}}
  \end{cases}
\end{equation*}
As above, the second condition is always satisfied. As for the first one, its r.h.s. can be maximized taking $x = 1-\delta$. Hence, a non-extreme consensus can be stable only if $\tb \leq \frac{\frac{\delta}{2}(1+h)}{1 - (1 - \delta) -\frac{\delta}{2}} = h + 1$, as desired.
\end{proof}
\end{theorem}

Next corollary highlights a fundamental difference between the discrete and the continuous setting. Indeed, stable profiles that are non-extreme consensus are feasible in the discrete case while it is known that they are not feasible in the continuous case.

\begin{corollary}
Given an $(\ain,\aout)$-symmetric two-block model $G=(L \cup R, E, w)$ and a social media influence $b$, if $\delta \rightarrow 0$, then consensus on opinions different from $-1,0,1$ can not be stable.
\end{corollary}
\begin{proof}
Suppose that consensus is achieved on $x \in \{\delta, \ldots, 1-\delta \}$ (the case in which consensus is achieved on a non-extreme negative opinion is symmetric and hence omitted).

 From the proof of Theorem 5 we know that, if $x \in \{\delta,\ldots, \lambda\}$, then consensus on $x$ can be stable only if $\tb \leq \frac{\frac{\delta}{2}(1+h)}{x-\frac{\delta}{2}}$.  When $\delta \rightarrow 0$, the condition became $\tb \leq 0$ that is impossible by definition.

 If $x \in \{\lambda + \delta,\ldots, 1-\delta\}$, from the proof of Theorem 5, we know that consensus can be stable only if $\tb \leq \frac{\frac{\delta}{2}(1+h)}{1-x-\frac{\delta}{2}}$. When $\delta \rightarrow 0$, the condition became $\tb \leq 0$ that is impossible by definition.
\end{proof}

\subsection{Convergence}
Previous sections characterized when and which consensus can be achieved depending on the homophily and the social media influence in opinion dynamics with synchronous update rules.

However, we should highlight that, although, as suggested above, a stable state always exists, synchronous updates can make the opinion dynamics unable to converge to this state from specific initial opinions profiles. In the following we provide a simple example.
\begin{proposition}
An opinion dynamics with synchronous updates may not converge to a stable state.
\begin{proof}
% We prove the preposition by defining a simple instance for which the synchronous dynamics does not converge.\\
Let consider a two-player social network. We note the two player as $l$ and $r$.  Let $w_{l,r} = 10$  the weight of the edge, $\Theta = \{-1, -\frac{1}{2}, 0, \frac{1}{2}, 1\}$ the opinions' set and $b = 1$ the strength of social media influence. The initial opinions of the agents are $x_l^0 = -1, x_r^0 = 1$.

At step $1$, the agent $l$ choose the opinion $x \in \Theta$ that minimizes $c_l(x^0) = 1(x+1)^2+10(x-1)^2$. It's easy to see why $x_l^1 = \argmin_{x}c_l(x^0) = 1$. In the same way, the opinion of agent $r$ at the step $1$ is $x_r^1 = \argmin_{x}c_r(x^0) = -1$. Consequently, at step $2$ we have $x_l^2 = 1, x_r^2 = -1$ and at step $t$ we have $x_l^t = -1^{t+1}, x_r^t = -1^{t}$. Therefore, the dynamics continues indefinitely since $x_l^{t-1} \neq x_l^{t} \text{ and } x_r^{t-1} \neq x_r^{t}, \forall t \geq 0$.
\end{proof}
\label{prep1}
\end{proposition}
However whenever the opinion dynamics with synchronous updates converges, the stable state at which it converges can be computed, since it deterministically depends only on the initial opinion profile $x^0 = (x_1^0, \ldots, x_n^0)$.
% \begin{observation}
% \label{obsv1}
% Given a social network $G = (V,E)$ and an initial opinion profile $x^0 = (x_1^0, \ldots, x_n^0)$, if the opinion dynamics is synchronous and it converges then the stable opinion profile is unique and it deterministically depends only on $x^0$.
% %\begin{proof}
% %For any player $i \in V$ the opinion at the step $1$ is:
% %\begin{equation*}
% %x_i^1 = \argmin_{x}{\{b(x-s(x_i^0))^2+\sum_{j:(i,j) \in E}w_{i,j}(x-x_j^0)^2\}}
% %\end{equation*}
% %Therefore, the opinion profile $x^1$  deterministically depends on $x^0$.\\
% %In the same way, the opinion profile $x_2$ only depends  on $x^1$ and so it only depends on $x^0$. Hence, for any step $t$, the opinion profile $x^t$ deterministically depends on $x^0$. Consequently, if the dynamics converges then the stable opinion profile only depends deterministically on $x^0$.
% %\end{proof}
% \end{observation}

In the section~\ref{seAsync} we will study the opinion dynamics with asynchronous updates. We will
show a specular behaviour: indeed, it will turn out that in this case the dynamics always converges to a stable state, but more than one stable states may exist.
Interestingly, in section~\ref{async_res} we will also show that the findings about when and which consensus are reached by the dynamics obtained through the analysis of the synchronous case even extend to the asynchronous case.

% Unfortunately, when synchronous update are allowed, the evolution of opinion may not converge to a stable state, as showed by the following example.
% 
% ESEMPIO QUI
% 
% For this reason, we next focus on the case of asynchronous updates, for which convergence always occurs, and as we will show, it often occurs in limited amount of time. Moreover, we will show that all our findings about probability of consensus in the case of synchronous updates surprisingly hold even for asynchronous updates.

%\paragraph{ALTRI RISULTATI SUL CONSENSO}
%AGGIUNGERE COMMENTI SU DELTA VA A 0 E H CHE VA A INFINITO\\
%\textcolor{red}{Forse non he senso discutere di questi limiti anche perchè non abbiamo fatto simulazioni in merito, sono osservazioni che poi non vengono usate nel resto del per}

\section{General Networks}
\label{sync_experiments}
In previous section we presented some results related to the two-block model of a social network. We conjecture that our results hold in more general settings under the hypothesis that it is possible to distinguish in the network two well separated sets of similar agents. In this section we present some experimental evidences to support our conjecture. In particular, we run our experiments on stochastic two-block model graphs, random graphs and on real graphs. In the latter two cases we use algorithmic techniques to separate nodes in two components and then we define weights of the edges in order to define the influence coming on an agent from her own component and from the other component. 
%We run our experiments for six classes of networks.

% \paragraph{The Networks used in our experiments.}
% % \smallskip \noindent \emph{Networks.}

Observe that in a symmetric two-block model network, for each agent $i$,
$\sum_{\begin{subarray}{c}j \in P\\(i,j) \in E\end{subarray}} w_{ij}$,
that is the influence that she receives from the other agents in the same component, is equal to a constant $\ain$. At the same time, for each agent $i$,
$\sum_{\begin{subarray}{c}j \in \notP\\(i,j) \in E\end{subarray}} w_{ij}$,
that is the influence that she receives from the other agents in the different component, is equal to a constant $\aout$. Moreover, agents in the same component have the same initial opinion. In our first experiment we extend this model by relaxing some of these assumptions. In particular, the set of vertices is $V = L \cup R$, where $\vert L\vert  = \vert R\vert  = N$; each edge between two agents in the same component exists with probability $p_{in}$ while each edge between agents in different components exists with probability $p_{out}$. However, all edges have the same weight. 
Thus, two agents in the same component may have different neighbours, even if they receive the same expected influences. Indeed, the expected influence received by her component is equal to $(N-1)p_{in}$ and the expected influence that an agent receive from her opposite component is equal to $N p_{out}$. Furthermore, agents in the same component can have different initial opinions.

We set $N = 50$ and simulate our opinion dynamics with different values of $\pin$, $\pout$, $\delta$, and $b$. For each setting we run $n_p = 1000$ simulations.
For each simulation, given the two blocks, say $L$ and $R$, we assume that for each agent in $L$, the initial opinion $x_i^0$ is drawn at random in the interval $[l_L, h_L]$, and for each agent in $R$ the choice is drawn at random in the interval $[l_R, h_R]$, where $h_L$ and $l_R$ are set respectively to $-\xi$ and $\xi$, where $\xi$ is drawn uniformly at random in the interval $[0, \lambda+\delta]$, $l_L$ is drawn at random in the interval $[-1, h_L]$, and $h_R$ is drawn at random in the interval $[l_R, 1]$. Let $m$ be the number of runs in which the dynamics converges to consensus, we measure the consensus probability as $p_c = m/n_p$ and the $95\%$ confidence interval as $p_c \pm2\sqrt{p_c(1-p_c)/n_p}$
%We measure the consensus probability defined as the fraction of the runs in which the dynamics converges to consensus.

Next we consider networks generated using three well-known network formation models: the Random Graphs model \citep{random}, the Watts-Strogatz model \citep{WattsStrogatz} and the Hyperbolic Random Graph model \citep{hyperbolic}. We remark that the Random Graph model is generally used to generate random networks. The other two graph models are known to generate networks enjoying properties usually more similar to real social networks. In particular, the Watts-Strogatz model is known to generate smallworld networks (i.e., network with small diameter and a large clustering index). In \citep{hyperbolic} it has been showed that the Hyperbolic Random Graph model, for a special choice of parameters, generate networks that are smallworlds with a degree distribution that is a power law, a characteristic that can be find in several real-life social networks.

In the Random Graphs model, for each pair of vertices $u$ and $v$ the edge $(u, v)$ is created with probability $p$. Notice that, in general, the random graph $G = (V, E)$ generated in this way cannot be separated in well-defined components of the same size. However, we can partition the set of vertices in two components, $L$ and $R$, by running the well-know algorithm of \cite{Kerningham-Lin}, that returns the partition generated by the sparsest cut. We then assign weight $\win$ to edges among nodes in the same component, and $\wout$ to all the remaining edges.
Note that, as in the two block model, here each node receives a different social influence from nodes within the same component and nodes of the opposite component. However, in this case even the expected influences received by agents of the same component may be different. Indeed, for each node $i \in P$, with $P \in \{L, R\}$, we set $\ain(i) = \win \cdot \vert \{j \in P \colon (i, j) \in E\}\vert $ and $\aout(i) = \wout \cdot \vert \{j \notin P \colon (i, j) \in E\}\vert $, and it is not hard to build a graph such that $\ain(i) \neq \ain(j)$ or $\aout(i) \neq \aout(j)$ for a pair of nodes $i,j$.
Nevertheless, we can define, even in this setting, the homophily ratio as $h = \frac{\ain^*}{\aout^*}$, where $\ain^* = \frac{1}{\vert V\vert } \sum_{v \in V} \ain(v)$ and $\aout^* = \frac{1}{\vert V\vert } \sum_{v \in V} \aout(v)$.

Watts-Stogatz networks are generated by positioning nodes in a metric (usually Euclidean) space and linking nodes through two classes of links: two nodes whose distance is below a given threshold $r$ are linked through so-called \emph{strong ties}; each node has $k$ additional links, termed \emph{weak ties}, connecting to randomly selected endpoints. 
Hyperbolic Random Graph networks are generated in a similar way but the metric space is restricted to be hyperbolic with negative curvature. For both these two classes of networks we partitionate nodes in two communities and set the weights of the links as described for the Random Graphs model.
For each of these three classes of graphs we set the number of agents to be $100$ and we run our opinion dynamics for different values of $\win$, $\wout$, $\delta$, and $b$. For each setting, we run 1000 simulations and we compute the initial opinions, the consensus probability and the confidence interval as stated before.

Finally, we considered two samples of real social networks that are freely available in the SNAP library \citep{snap}. The first one, \texttt{ego-Facebook}, is a sample of
4039 nodes and 88234 edges retrieved from 
% a portion of the
Facebook
network \citep{ego-Facebook}. The second one, \texttt{\seqsplit{feather-lastfm-social}} consists in a less dense network
of 7624 nodes and 27806 edges
\citep{lastfm}. In order to run multiple simulations on these networks we do not use a deterministic partitioning algorithm to retrieve communities, but for each simulation an agent is assigned to cluster $L$ with a probability $p_L$ drawn uniformly at random in $[0.4, 0.6]$, and to cluster $R$ otherwise. We will show below that, despite this random choice of the partitions, we still are able to achieve results that are similar to previous more regular networks. Weights, homophily ratio, the influence of the social media, initial opinions, consensus probability and the confidence interval are then computed as described above (but mediated over only 500 simulations, due to the larger size of these networks).

We observe that numerical oscillations can make impossible to reach consensus even if opinions of agents are very close to each other. For this reason, we consider a \emph{relaxed} definition of \emph{consensus}. In particular, following the analysis of \cite{immorlica}, we will focus on the average opinion $\overline{x}_P = \frac{\sum_{i \in P} x_i}{\lvert P \rvert}$ for each partition $P$, and on its projection $\underline{x}_P$ on $\Theta$, being the opinion in $\Theta$ closest to $\overline{x}_P$. Then, a stable opinion profile $x = (x_1, x_2, \ldots, x_n)$ is a consensus if $\underline{x}_L = \underline{x}_R$. We performed extensive simulations to determine if the \emph{relaxed} definition of consensus can make our results inaccurate. Specifically, we compute the distributions of opinions when a relaxed consensus is reached. We observe that when a relaxed consensus is reached, the mean of the opinions is near to the relaxed consensus value and the variance is lower than $0.5 \cdot \delta$. Remember that the opinion of each agent belongs to the set $\Theta = \{-1, -1+\delta, \ldots, 0, \ldots, 1-\delta, 1\}$. In Figure~\ref{loss_in_Consensus} we show the value of the relaxed consensus and the distribution of opinions at relaxed consensus on the \texttt{ego-facebook} social network, as $\tb$ increases. In particular, we can observe that the mean is near to the relaxed consensus value and the variance is tight w.r.t the value of $\delta$. The results showed are confirmed by repeating the same simulations $30$ times. We obtained very similar results for different values of $\delta$ and in most of the models considered in our analysis: Two-Block Model, Random graph, Watts-Strogatz graphs, Hyperbolic Random graphs and \textit{feather-lastfm-social}. We omit these other results for the seek of readability.

\begin{figure}[htb]
\centering
\includegraphics[scale=0.6]{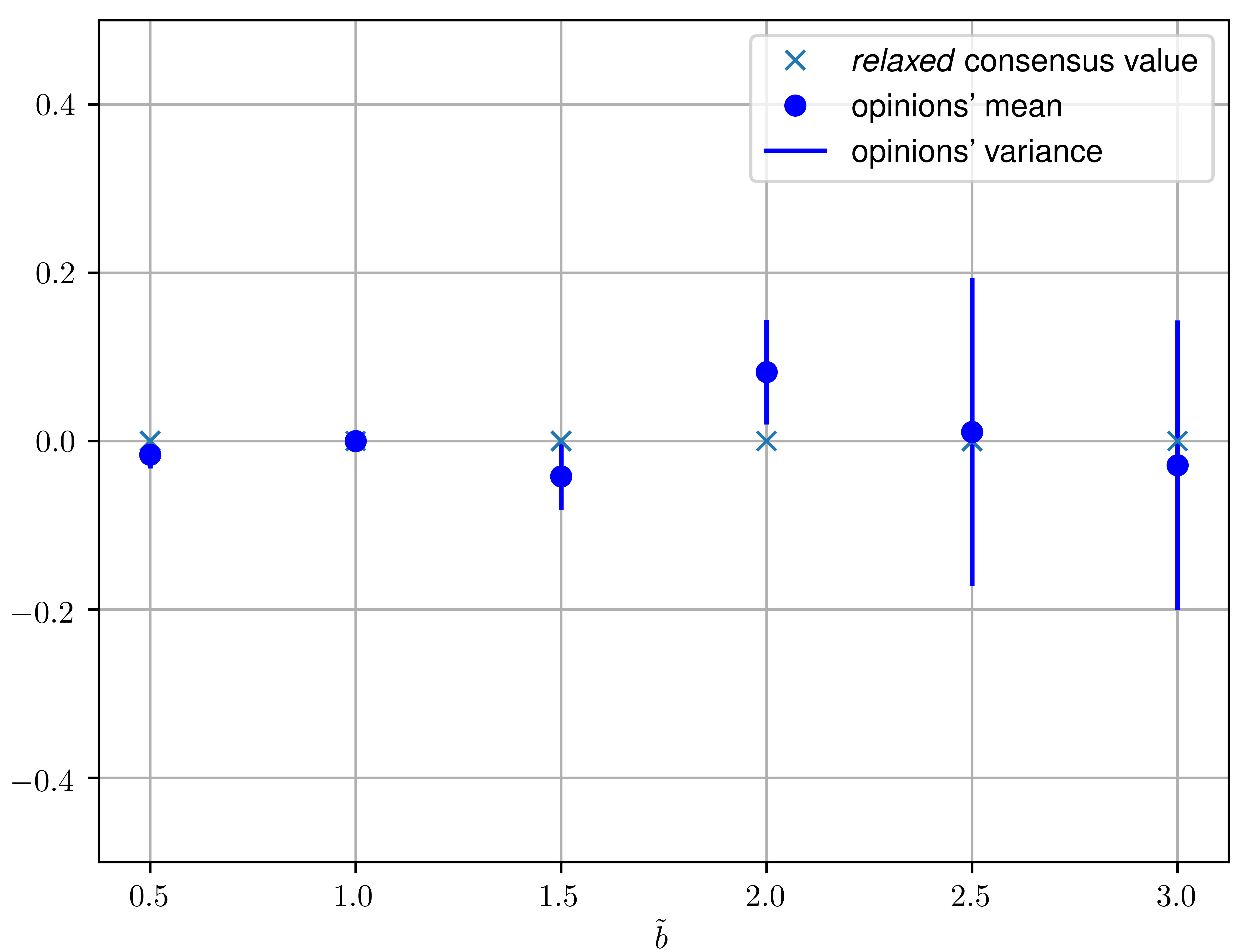}
\caption{
We show the value of \emph{relaxed} consensus, the opinion's mean and the opinion's variance at the \emph{relaxed} consensus for \texttt{ego-facebook} social network when $\delta = 0.5$ and $b$ is set to have $\tb \in \{0.5, 1, 1.5, \ldots, 3\}$.\\
Similar results are obtained for different value of $\delta$ and different network models.
} 
\label{loss_in_Consensus}
\end{figure}

In the light of these results, in the following we consider the \emph{relaxed} consensus since it is numerically stabler than consensus and preserves the accuracy of the results. We will refer to the \emph{relaxed} consensus simply as consensus.

Note that the stochastic networks strongly depend on the particular choice of the input parameters of the randomized generative algorithm. In particular, the Random Graph model strongly depends on the probability ($p$) of creating an edge between each pair of nodes of the graph, while the Watts-Strogatz model depends on the particular choices of $r$ and $k$. We observed that our results are essentially independent from these particular choices of the parameters. Indeed, we are interested on the trend of the consensus probability and not on the absolute values. We observed that the former does not change by varying the model's parameters and it reflects the findings discussed in section~\ref{sec3}.
For example, in Figure~\ref{fig_p} we show the trend of the consensus probability for the Random Graph Model for different values of $p$. In particular, we show how the trend of the consensus probability changes as the social media influence (Figure~\ref{fig_p_a}), the initial opinion of agents (Figure~\ref{fig_p_b}), and the homophily ratio (Figure~\ref{fig_p_c}) changes. In Figure~\ref{fig_p_c} we observe that there are not negligible numerical differences between $p = 0.6$ or $p=0.8$ and $p=0.2$, and they increase as the average initial opinion of agents in a component increases. Specifically we have a difference larger than $0.2$ when $\tb > 0.8$. However, as we specified above, we are interested on the consensus probability. We can observe that, while we have a change of scale, the consensus probability follows the same monotone trend for all the density values $p$ of the Random Graph.

\begin{figure}[htb]
	\begin{subfigure}[c]{0.49\textwidth}
\includegraphics[scale=0.415]{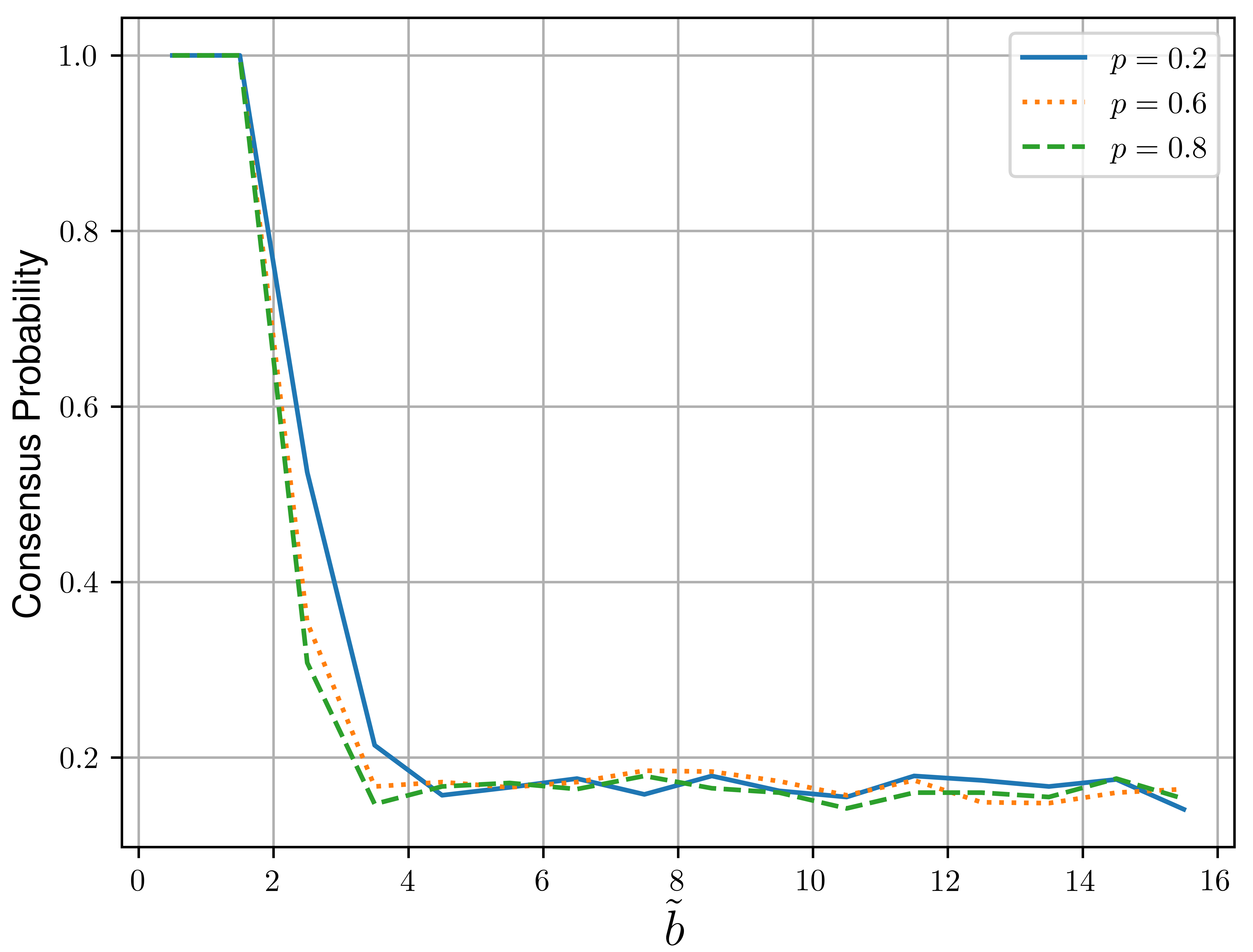}
\caption{Social media's influence}
\label{fig_p_a}
\end{subfigure}
\hfill
	\begin{subfigure}[c]{0.49\textwidth}
\includegraphics[scale=0.415]{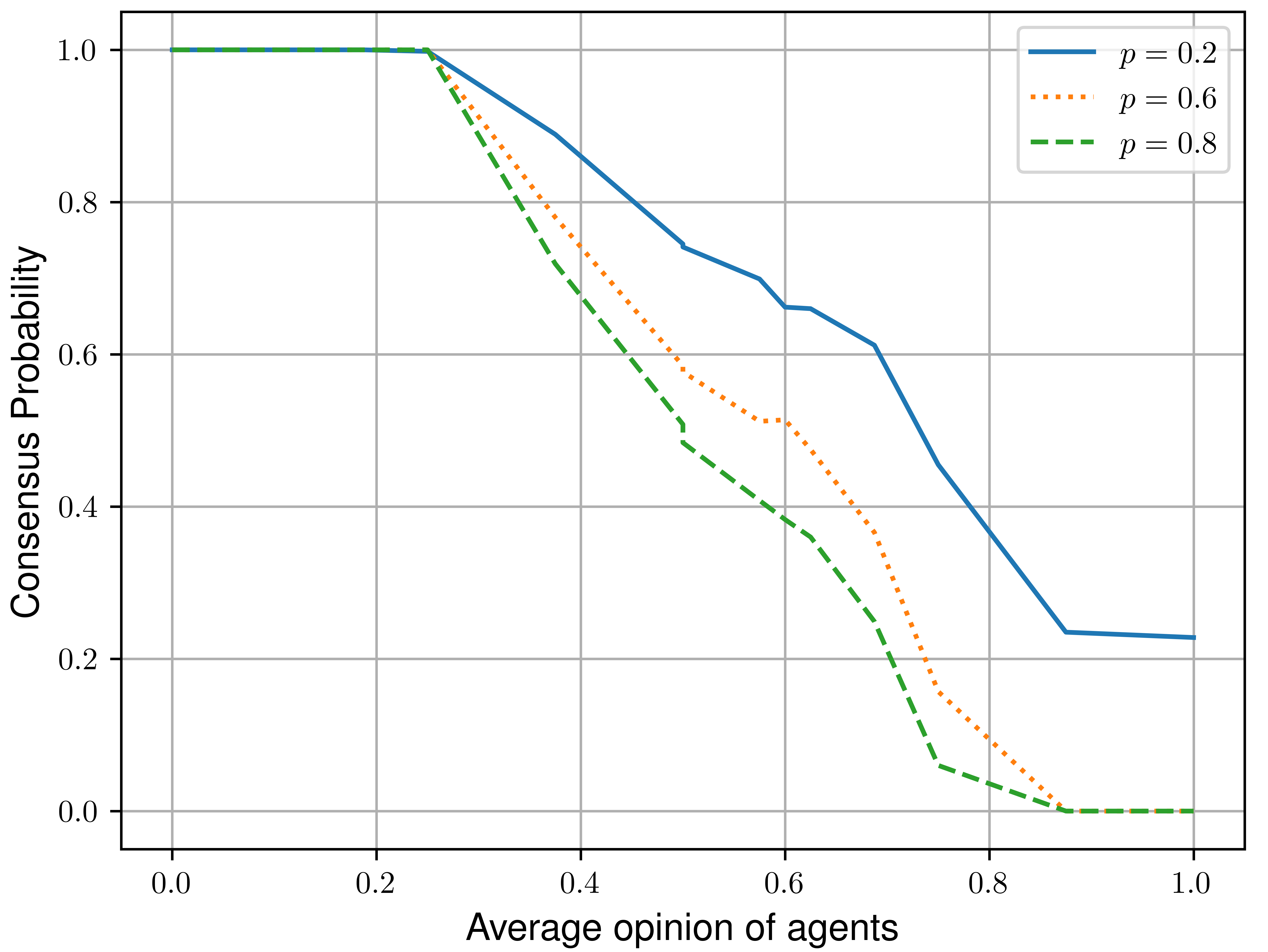}
\caption{Initial opinions}
\label{fig_p_b}
\end{subfigure}\\
	\begin{subfigure}[c]{\textwidth}
\centering
\includegraphics[scale=0.415]{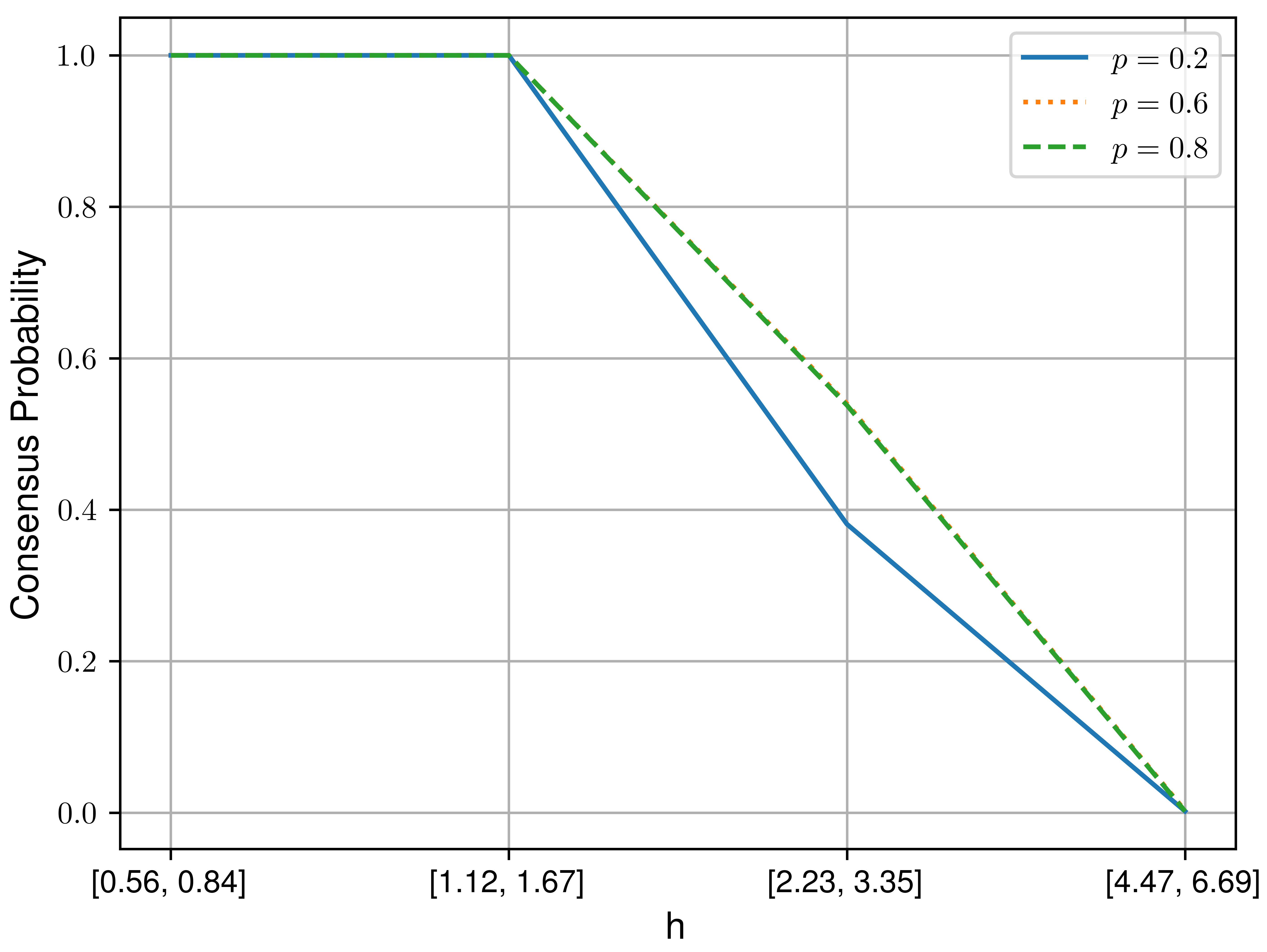}
\caption{Homophily}
\label{fig_p_c}
\end{subfigure}
\caption{
We show the trend of the consensus probability for the Random Graph model for $p \in \{0.2, 0.6, 0.8\}$.\\
In Figure~\ref{fig_p_a} we show the trend of the consensus probability when $\delta = 0.125, \win = \wout = 1$, and $b$ is set to have $\tb \in \{0.5, 1.5, \ldots, 14.5\}$.\\
In Figure~\ref{fig_p_b} we show how the consensus probability changes as the average opinion of agents in a component varies when $\delta = 0.25, \win = 1, \wout$ is drawn at each simulation uniformly at random in $[0.3, 4]$ and $\tb = 1.2\tau_1(h)$.\\
In Figure~\ref{fig_p_c} we show the consensus probability trend when $\delta = 0.125, \tb = 0.5, \wout = 1$ and $\win = \{0.5, 1, 2, 4\}$. Notice that, $h$ is the expected value of the homophily ratio over all runs involving the same value of $\win$. 
} 
\label{fig_p}
\end{figure}

For the Watts Strogatz Model we observe almost the same results. In particular, we computed how the trend of the consensus probability varies as the influence of social media, the opinion of agents and the homophily ratio changes for different combination of parameters. We omit these results for the seek of readability since they are very similar to the ones showed in Figure~\ref{fig_p}.

We shown that the trend of the consensus probability does not depends on the particular choice of the parameters of the generative algorithms. In the light of this, we decided to do not assume any specific value for these parameters.
% since the results discussed before assure us that we are not losing any crucial informations or dependencies about the trend of the consensus probability.
Specifically, for Random Graph networks, at each simulation we draw the value of $p$ uniformly at random in the interval $[0.3, 0.7]$. Similarly, for Watts-Strogatz networks, at each simulation we draw $r$ uniformly at random in the interval $[5, 40]$ and $k$ uniformly at random in the interval $[4,30]$.

On the other hand, for the Hyperbolic Random model we adopted a different approach. We fixed some of the parameters of the generative algorithm to make the networks similar to real social networks \citep{hyperbolic}. Specifically, the Hyperbolic Random Model takes three parameters: the exponent of the power-law distribution $\gamma$, the temperature $T$ and the target average degree of each node $K$. In order to have a coefficient of clustering significantly higher than $0$ and an efficient decentralized search, properties frequently observed in real social networks, we must set $T < 1$ and $\gamma < 3$. Specifically, we set $T = 0.6$ and $\gamma = 2.5$. Moreover, we observed that the trend of consensus probability does not significantly change by varying the value of $K$, we omit these results for the seek of readability since they are very similar to the one showed in Figure~\ref{fig_p}. Hence, as we did for the other network formation models, at each simulation we draw $K$ uniformly at random in the interval $[0.07 \cdot \lvert V \rvert, 0.75 \cdot \lvert V \rvert]$.

\paragraph{Our Experimental Results}
Our experiments highlight that the consensus probability essentially depends only on the relative social media influence $\tb$ and on the homophily ratio $h$, and not on the absolute values of the media influence $b$, the inter-cluster influence $\aout$, and the intra-cluster influence $\ain$. Indeed, Figure~\ref{fig1} shows that the probability of reaching a consensus is essentially the same when $\tb$ and $h$ are unchanged, even if we change the values of $b, \aout, \ain$.
\begin{figure}[htb]
% \subcaptionbox{Two Block Model -- ToUpdate}[0.5\linewidth][c]{\includegraphics[scale=0.41]{/ConsensusProbabilityCompAlpha/BlockModel}}
\begin{subfigure}[c]{0.49\linewidth}
% \centering
\includegraphics[scale=0.4]{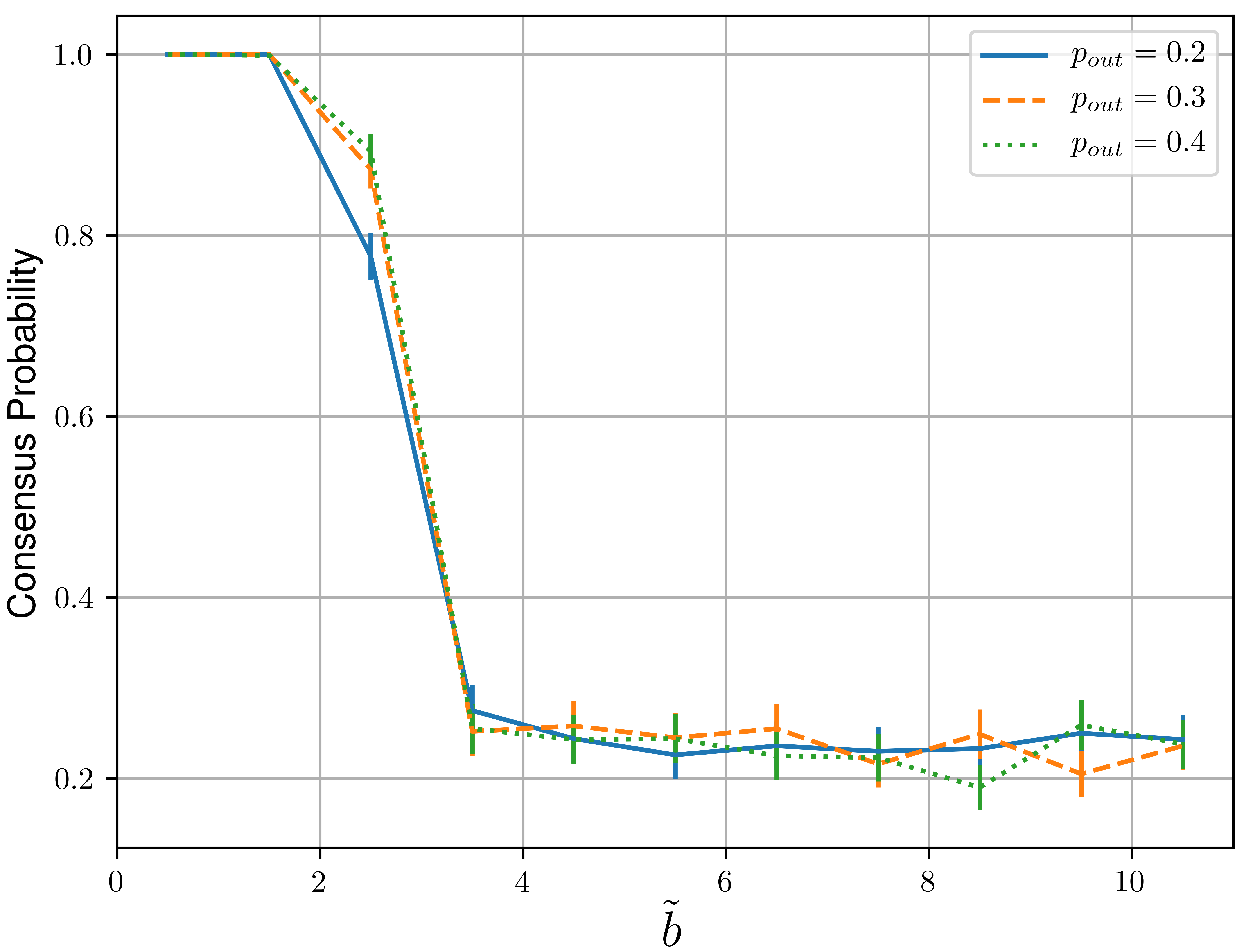}
\caption{Two Block Model}
\label{fig1a}
\end{subfigure}%
\begin{subfigure}[c]{0.49\linewidth}
% \centering
\includegraphics[scale=0.4]{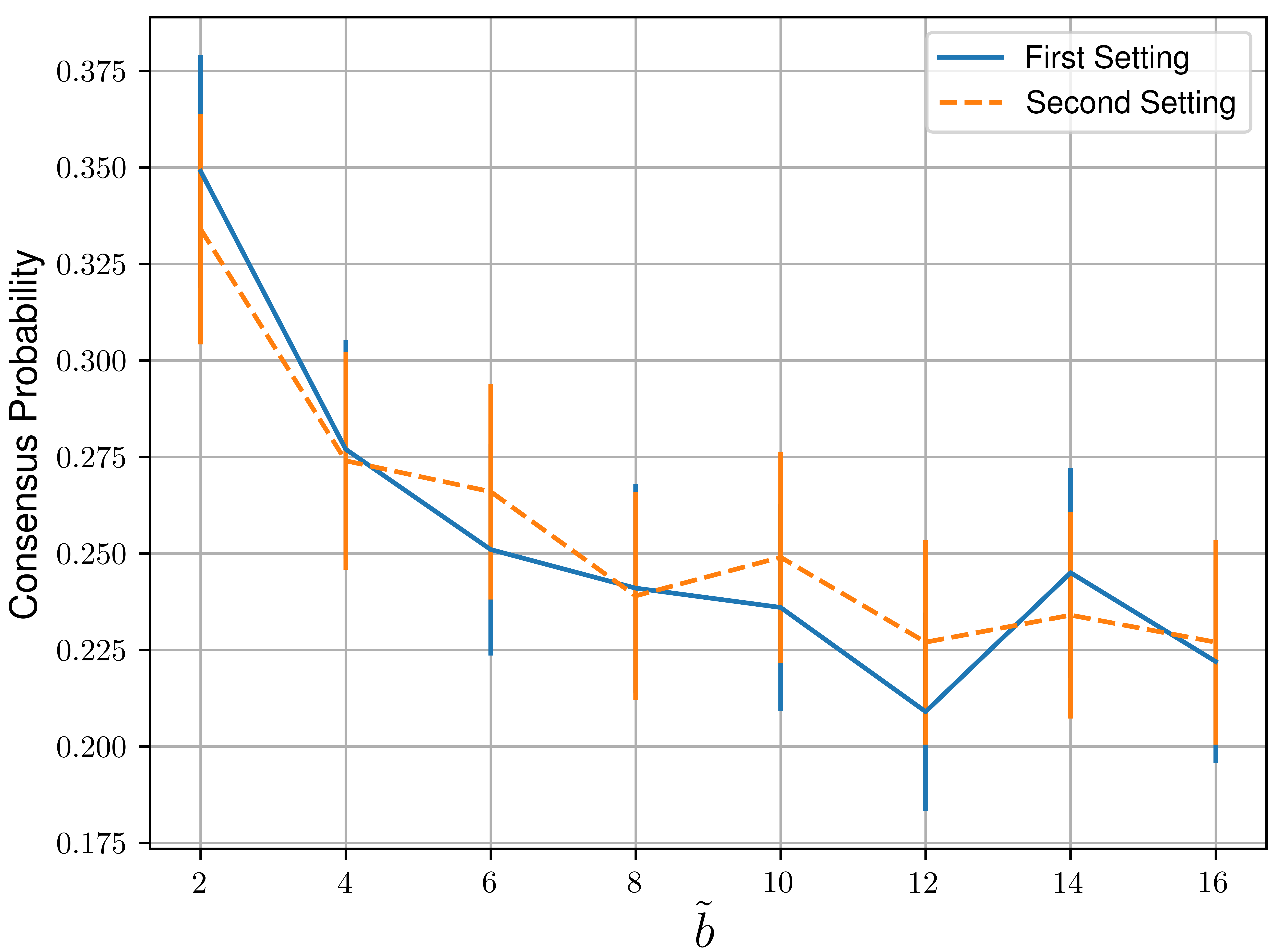}
\caption{\texttt{ego-Facebook}}
\label{fig1b}
\end{subfigure}
 \caption{In Figure~\ref{fig1a} we show how the consensus probability changes when  $\delta = 0.25$, $\pin = \frac{100\pout}{ 49}$ (so $h$ is always $2$) and $b \in \{25\pout, 25\pout+50\pout, \ldots, 525\pout)\}$ (so $\tb \in \{0.5, 0.5+1, \ldots, 10.5\}$).\\
 In Figure~\ref{fig1b} we show the consensus probability for $\tb \in \{0.5, 2.5, \ldots, 14.5\}$, and for each of these value, we evaluate this probability on two different settings.
%  (that will be designed to have the same value for the homophily ratio):
 in the first,
%  setting, we assume 
$\win = \wout = 1$
% (that will correspond to a value $h_1$ for the homophily ratio)
(from which the homophily ratio is $h_1$)
and $b$ such that $b\aout^* = \tb$;
% , where $\aout^*$ is computed as described above;
in the second,
% setting, we consider a 
with different
% partition in 
blocks,
we set
% and assume
$\win = 1$, $\wout$ is such that the
% corresponding
homophily ratio
% $h_2$
is
% exactly the same as 
is equal to $h_1$, and $\tb$
is
as in the first setting.\\
% , is chosen in order to achieve the desired value for $\tb$.\\
 %
The confidence intervals are shown as error bars. Similar results have been obtained for the other network models.}
%  and are  due to space limits.}
\label{fig1}
\end{figure}

The analysis of the symmetric two-block model also highlights that the probability of consensus usually decreases when either $\tb$ or $h$ increases. This behaviour is confirmed in all our experiments, even for the more complex networks. Specifically, Figure~\ref{fig2} shows how the probability of consensus changes as $\tb$ increases for different values of $h$.
\begin{figure}[htb]
 \begin{subfigure}{0.49\textwidth}
\centering
\includegraphics[scale=0.415]{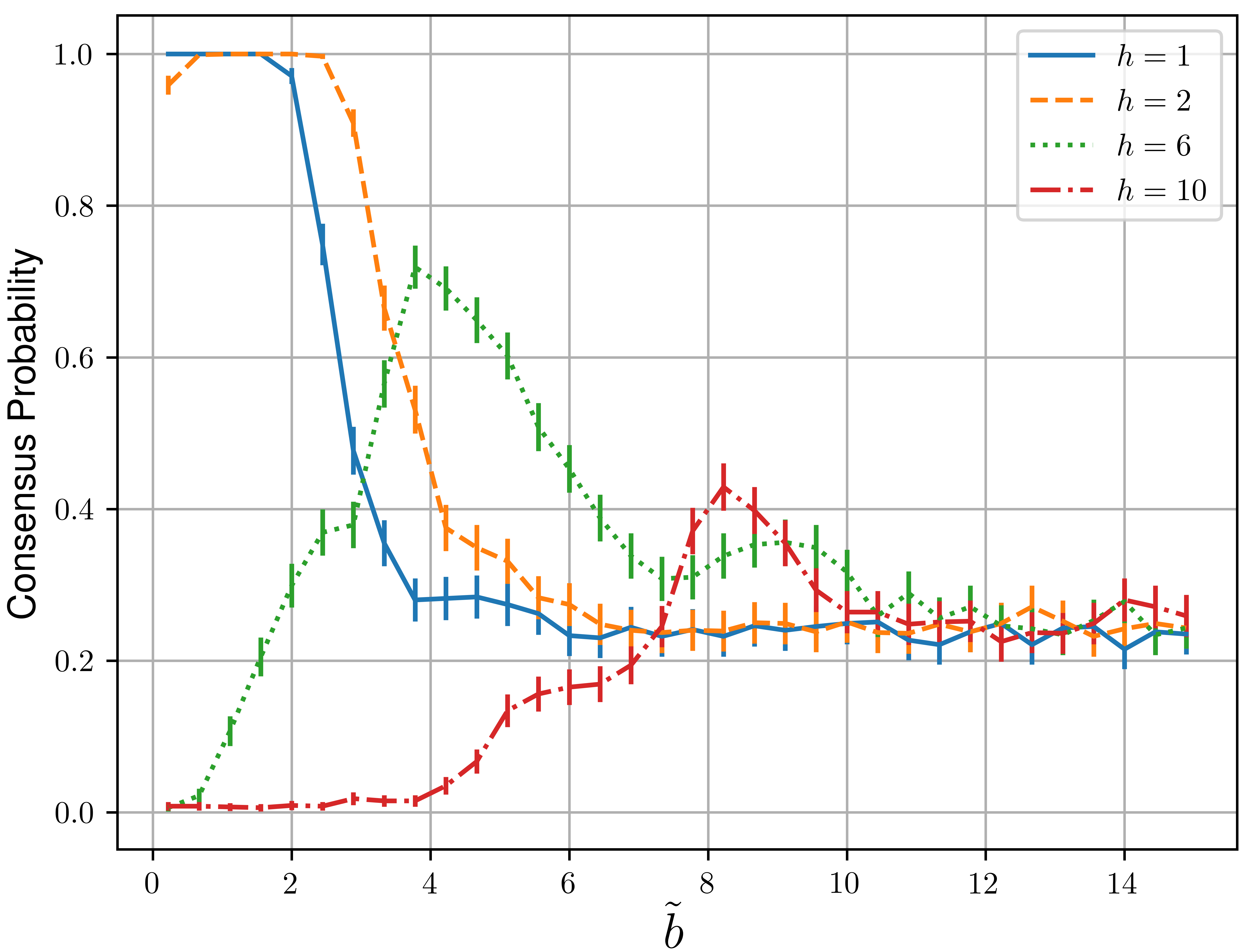}
\caption{Two Block Model}
\label{fig2a}
\end{subfigure}
\hfill
\begin{subfigure}{0.49\textwidth}
\centering
\includegraphics[scale=0.415]{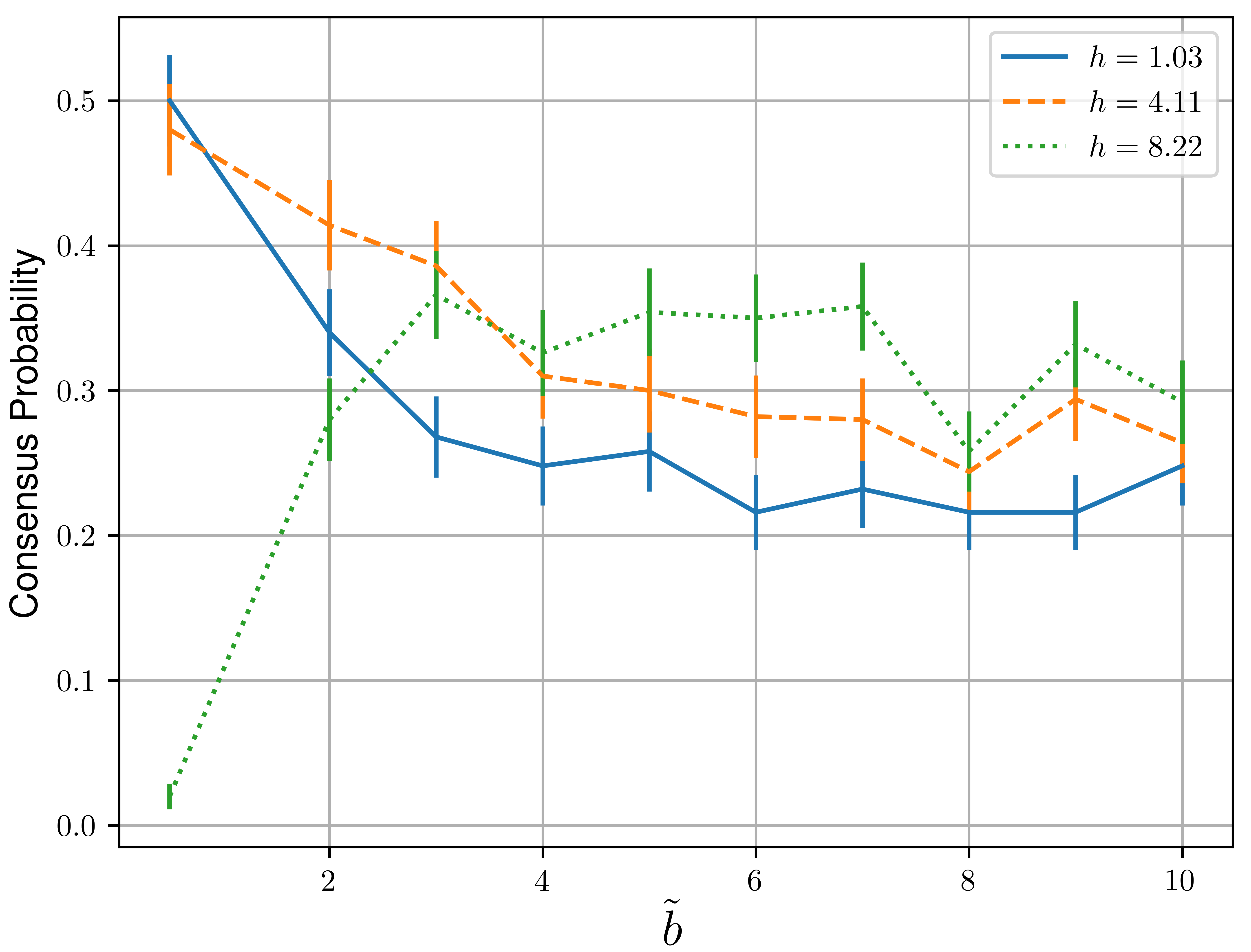}
\caption{\texttt{ego-Facebook}}
\label{fig2b}
\end{subfigure}
 \caption{In Figure~\ref{fig2a} we show how the consensus probability changes when $\delta = 0.25$, $\pout = 0.09$, $\pin$ is set in order to have the desired $h$, and $b$ is set to have $\tb \in \{0.3, 0.3+\frac{1}{25\pout}, \ldots, 14.88\}$).\\
 In Figure~\ref{fig2b} we show how the consensus probability changes when $\delta = 0.25$, $\wout = 1$, $\win \in \{1, 4, 8\}$, and $b$ is set to have $\tb \in \{1,2, 3, \ldots, 10\}$). Note that, in this setting, $h$ is the expected value of the homophily ratio over all runs involving the same value for $\win$.\\
The confidence intervals are shown as error bars. Similar results have been obtained for the other network models.}
%  and are omitted due to space limits.}
\label{fig2}
\end{figure}
% \begin{figure}[htb]
% %  \includegraphics{}
%  \caption{The figure refers to simulation with {\color{blue} two-block model}. Similar results are achieved also with different graph models. Bars provide the 95\% confidence interval.}
%  \begin{subfigure}[t]{0.45\textwidth}
% \centering
% \includegraphics[scale=0.45]{/ConsensusProbabilityCompAlpha/BlockModel}
% \caption{two-block nodel: trend of the consensus probability as $\tb$ increases for different value of $h$}
% \label{fig:chap_5_symmetric_stochastic_block_model_balanced_1}
% \end{subfigure}
% \hfill
% \begin{subfigure}[t]{0.45\textwidth}
% \centering
% \includegraphics[scale=0.45]{/ConsensusProbabilityCompAlpha/Random}
% \caption{random graphs: trend of the consensus probability as $\tb$ increases for different value of $h^*$}
% \label{fig:chap_5_symmetric_stochastic_block_model_balanced_2}
% \end{subfigure}
% \newline
% \centering
% \begin{subfigure}[t]{0.45\textwidth}
% \centering
% \includegraphics[scale=0.45]{/ConsensusProbabilityCompAlpha/LastFm}
% \caption{feather-lastfm-social: trend of the consensus probability as $\tb$ increases for different value of $h^*$ ({\color{red}questa sono i risultati di async, gli esperimenti su sync ancora non terminano le esecuzione)}}
% \label{fig:chap_5_symmetric_stochastic_block_model_balanced_3}
% \end{subfigure}
% \label{fig2}
% \end{figure}
It is immediate to see that, except for low values of $\tb$, the probability of consensus effectively decreases with $\tb$. Moreover, our results show that, for each value of $\tb$, the probability of consensus usually appears to be lower when $h$ is large (notice that, due to the fact that for very large $\tb$ the probability of consensus is very small, in this range the results showed in Figure~\ref{fig2} are highly affected by statistical noise, as it is also highlighted by the fact that the 95\% confidence interval are much larger in this range). An apparently strange behaviour occurs for low values of $\tb$. Indeed, in this range we have that the consensus probability increases. However, this behaviour is still in line with the theoretical results achieved for the symmetric two-block model.
Indeed, as observed above, for large values of $h$, the probability of consensus is expected to have this non-monotone behaviour: it first increases (by going from no consensus to possible consensus on $0$), and then decreases (by going from possible consensus on $0$ to no consensus again).

Results in Section~\ref{sec3} show that convergence to consensus is affected by the initial opinions of agents: indeed conditions for non-consensus in case both initial opinions are larger than $\lambda$ in absolute value are stricter than in the case of a single initial opinion far from $0$, and the latter are much more stricter in the case of both initial opinion are close to zero. This behaviour still holds even in more complex graph structures. Specifically, Figure~\ref{fig3} shows that the consensus probability decreases as the average opinion of the agents in a component goes to $1$.
\begin{figure}[htb]
  \begin{subfigure}[t]{0.49\textwidth}
\centering
\includegraphics[scale=0.415]{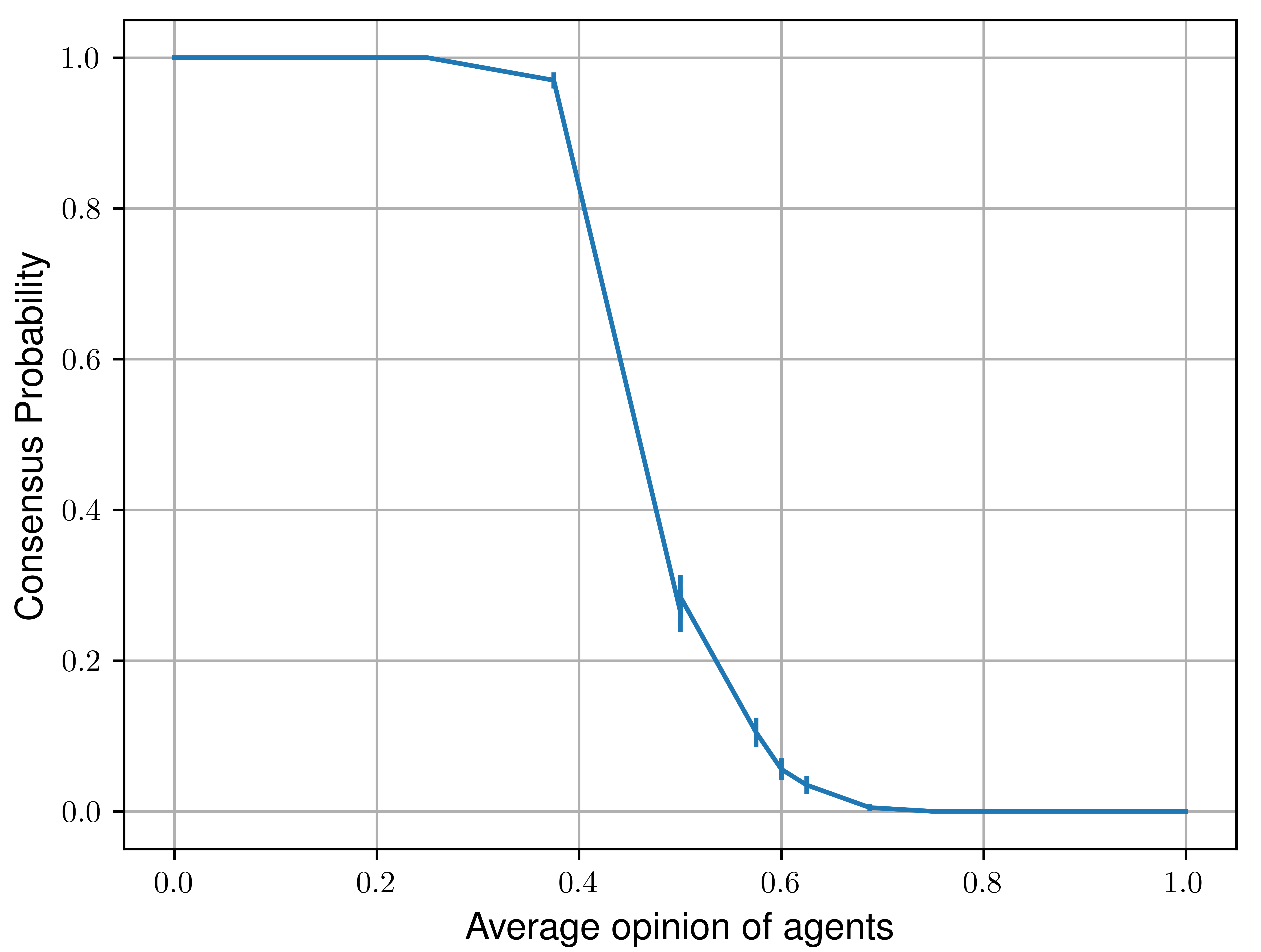}
\caption{Two Block Model}
\label{fig3a}
\end{subfigure}
\hfill
\begin{subfigure}[t]{0.49\textwidth}
\centering
\includegraphics[scale=0.415]{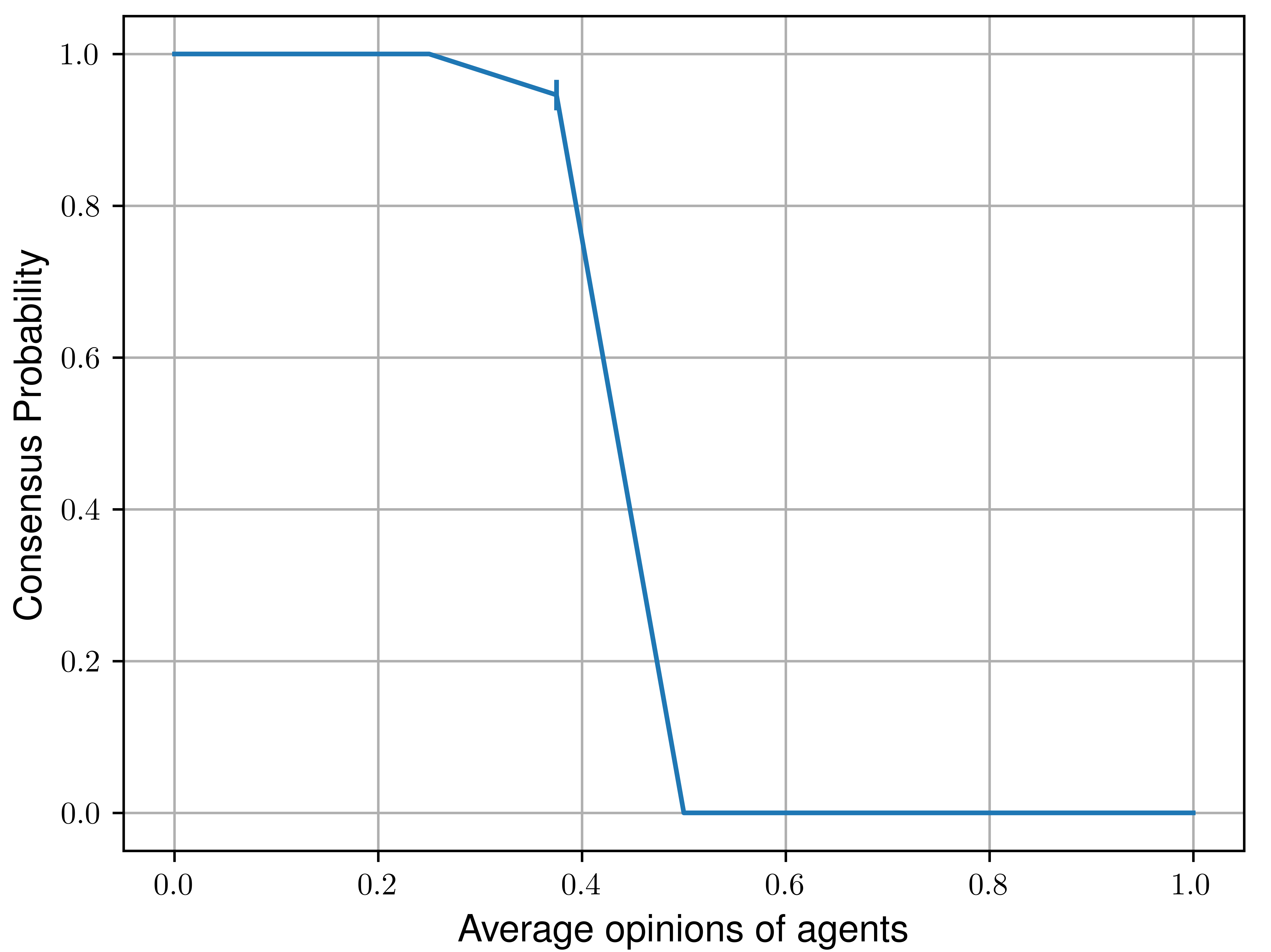}
\caption{\texttt{ego-facebook}}
\label{fig3b}
\end{subfigure}
 \caption{In Figure~\ref{fig3a} we show how the consensus probability changes when $\delta = 0.25$, $\pout = 0.2$, $\pin = 0.4$, and $b$ is set to have $\tb = 1.2\cdot\tau_1(h)$.
 In Figure~\ref{fig3b} we show how the consensus probability changes when $\delta = 0.25$, $\win = 1$, $\wout$ is drawn at each simulation uniformly at random in $[0.3, 4]$, and $b$ is set to have $\tb = 1.2\cdot\tau_1(h)$.\\
 Note that in this experiment the initial opinions are still drawn uniformly at random in intervals $[l_L, h_L]$ and $[l_R, h_R]$, but these interval are fixed (they are chosen to be the same interval but with opposite sign) to have that the average opinions of agents in each component has absolute value $x\delta$, with $x \in \{0, 0.6, 0.8, 1, 1.5, 2,3, 4\}$.\\
The confidence intervals are shown as error bars. Similar results have been obtained for the other network models.}
\label{fig3}
\end{figure}
Interestingly, the figure highlights that a sharp change of probability occurs exactly around $\lambda=0.5$, by confirming our findings.

Actually, we also run experiments in which we impose initial opinions to be larger or smaller than $\lambda$. Again, we observe that the behaviour on complex networks is very close to the one described for the simple symmetric two-block model.
Specifically, we first considered the case in which initial opinions are restricted to be larger than $\lambda$. Figure~\ref{fig4} shows how the probability of consensus changes as $\tb$ increase in this setting.
\begin{figure}[htb]
\begin{subfigure}[c]{0.49\textwidth}
\includegraphics[scale=0.415]{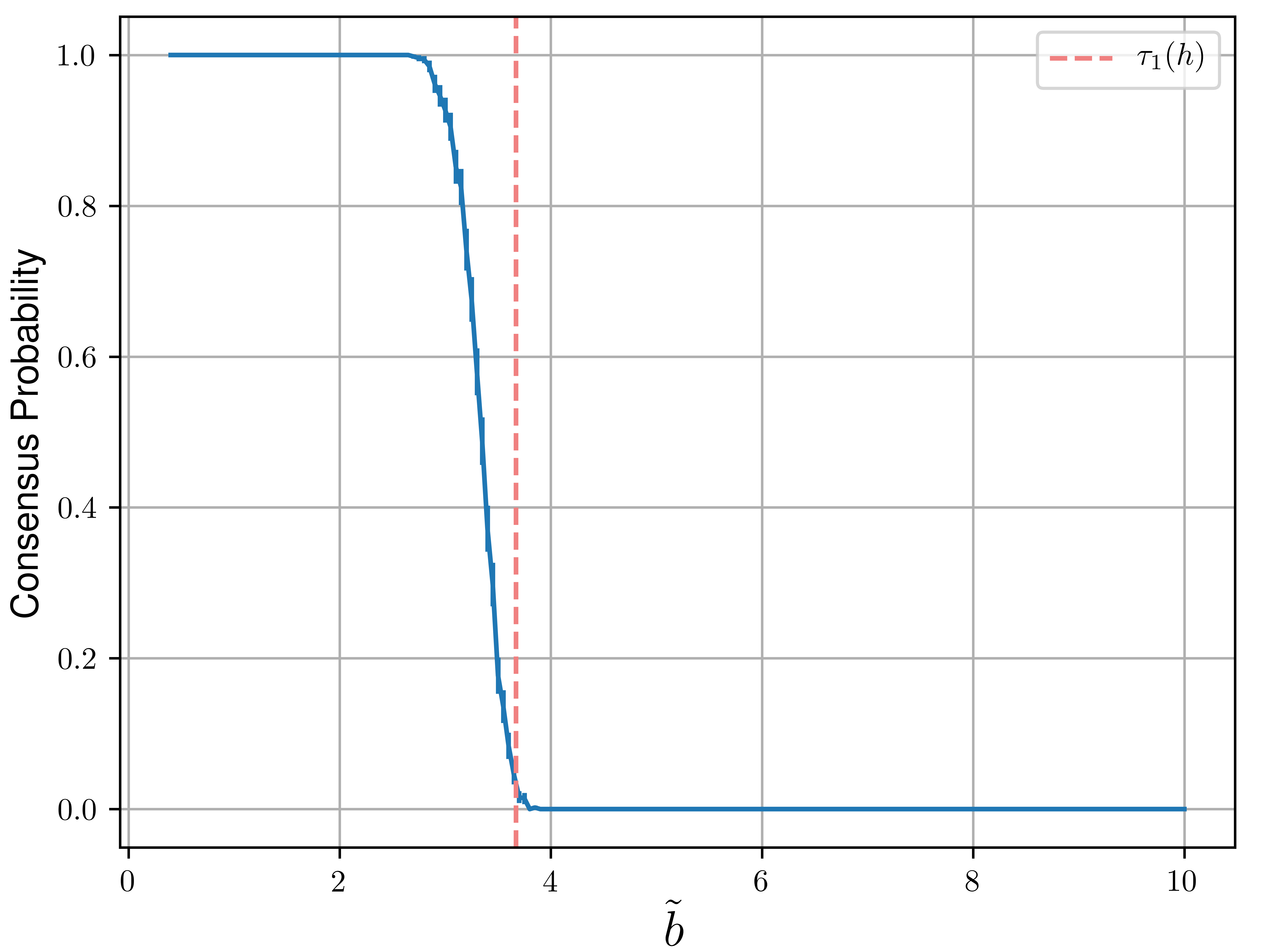}
\caption{Two Block Model}
\label{fig4a}
\end{subfigure}%
\begin{subfigure}[c]{0.49\textwidth}
\includegraphics[scale=0.415]{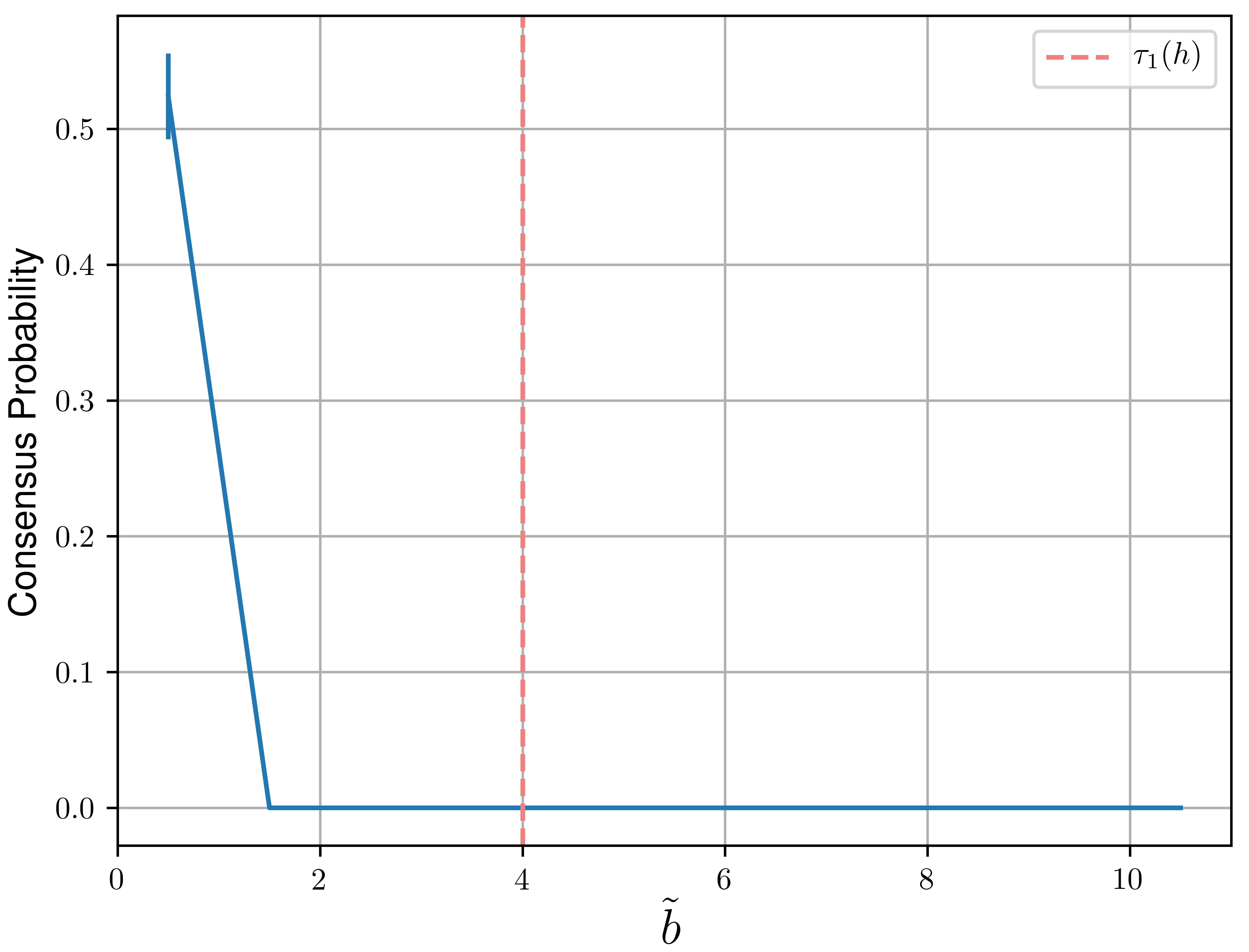}
\caption{\texttt{ego-facebook}}
\label{fig4b}
\end{subfigure}
\caption{In Figure~\ref{fig4a} we show how the consensus probability changes when $\delta = 0.25$, $\pout = 0.2$, $\pin = 0.4$, and $b$ is set to have $\tb \in \{0.4, 0.405, 0.41, \ldots, 10\}$.\\
 In Figure~\ref{fig4b} we show how the consensus probability changes when $\delta = 0.25$, $\win = \wout = 1$, and $b$ is set to have $\tb \in \{0.5, 1.5, \ldots, 10.5\}$.\\
 In order to have large divergent initial opinion, these are drawn uniformly at random in intervals $[-1, -\lambda-\delta]$ and $[\lambda+\delta, 1]$.\\
The confidence intervals are shown as error bars. Similar experiments have been run also on the remaining network models with very similar results.}
\label{fig4}
\end{figure}
We observe that the figure shows that the decrement of the probability of consensus occurs as soon as the value of $\tb$ is around $\tau_1(h)$. Hence, not only the general behaviour that emerges from the symmetric two-block model extends to more general networks, but we can also say that the given thresholds turn out to be quite precise in describing the behaviour also in more general networks.

Similar observations hold when we consider that agents are allowed to take opinions smaller than $\lambda$ in absolute values (see Figure~\ref{fig5}).
\begin{figure}[htb]
 \begin{subfigure}[c]{0.49\textwidth}
\includegraphics[scale=0.415]{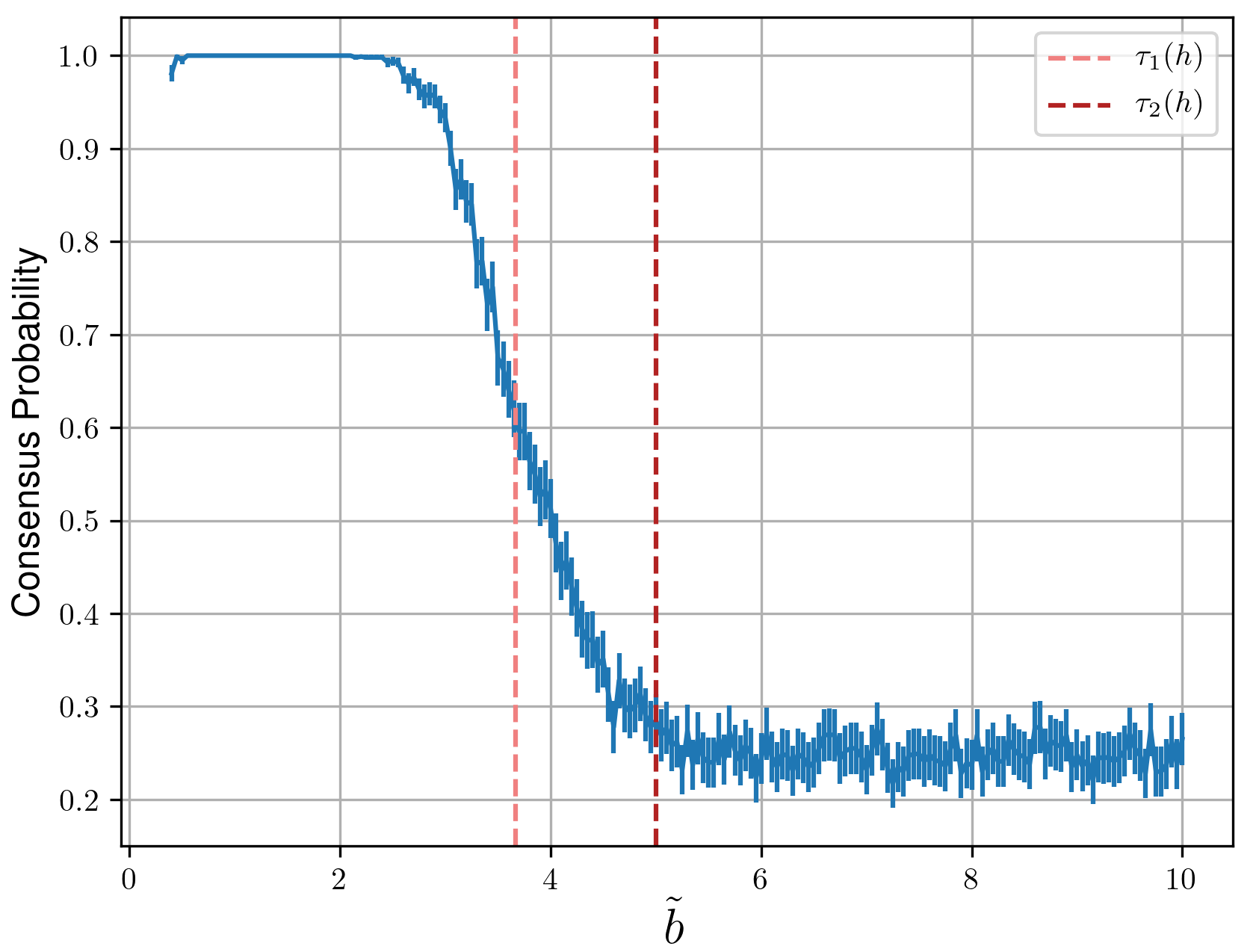}
\caption{Two Block Model}
\label{fig5a}
\end{subfigure}
\hfill
\begin{subfigure}[c]{0.49\textwidth}
\includegraphics[scale=0.415]{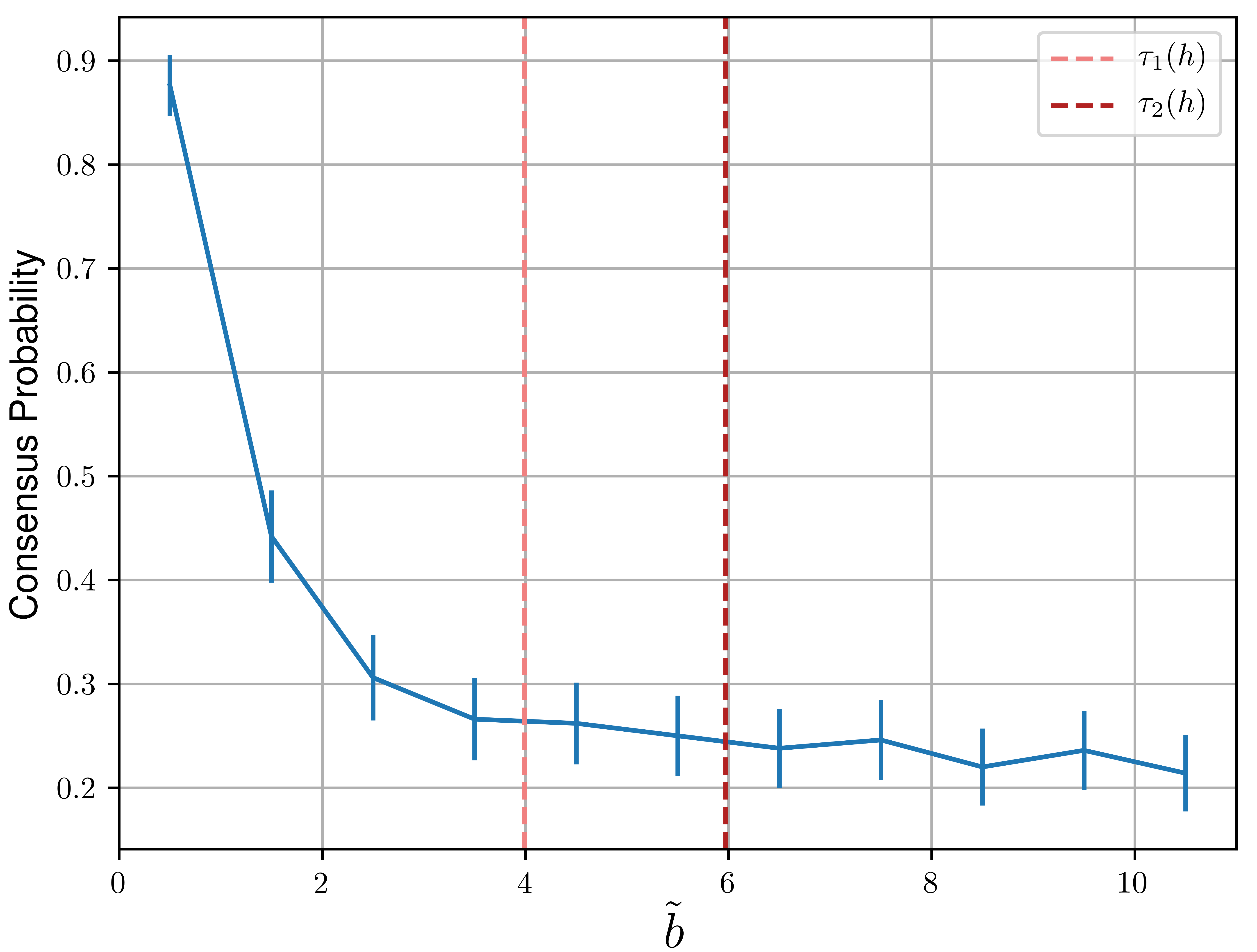}
\caption{\texttt{ego-facebook}}
\label{fig5b}
\end{subfigure}
 \caption{Figures show the results of experiments in the same settings as Figure~\ref{fig4}, except that initial opinion are not constrained to be far from zero (i.e., they are generated as described above).\\
The confidence intervals are shown as error bars. Similar experiments have been run also on the remaining network models with very similar results.}
\label{fig5}
\end{figure}
Here, there are three possible phases: when $\tb$ is small, we have high probability of consensus; for intermediate values, the probability is smaller, but still far away from zero; finally, for large $\tb$, the probability of consensus get close to zero. Interestingly, the phases changes occurs, as indicated by results above, around $\tau_1(h)$ and $\tau_2(h)$.

Finally, we consider the case in which initial opinions do not diverge.
Specifically, Figure~\ref{fig6} shows that, except for real datasets, it is possible to distinguish two phases: for low values of $\tb$, there is an high probability of consensus, whereas for larger values this probability decreases (but it does not go to zero).
Moreover, Figure~\ref{fig6} confirms that the smaller is the homophily ratio the smaller $\tb$ need to be to make the probability of consensus large by allowing consensus on an opinion different from $-1, 0, 1$, as observed in Theorem~\ref{thm:conv_small}.
Interestingly, for real datasets this behaviour is not confirmed, since the probability of consensus is close to $1$ regardless the values of $\tb$ and $h$. We conjecture that this different behaviour depends on the large density of these networks with respect to the remaining ones. However, this will requires a more careful analysis that would focus on the link between the impact of social media recommendations and the structural and topological properties of the network.
\begin{figure}[htb]
 \begin{subfigure}[c]{0.49\textwidth}
\includegraphics[scale=0.415]{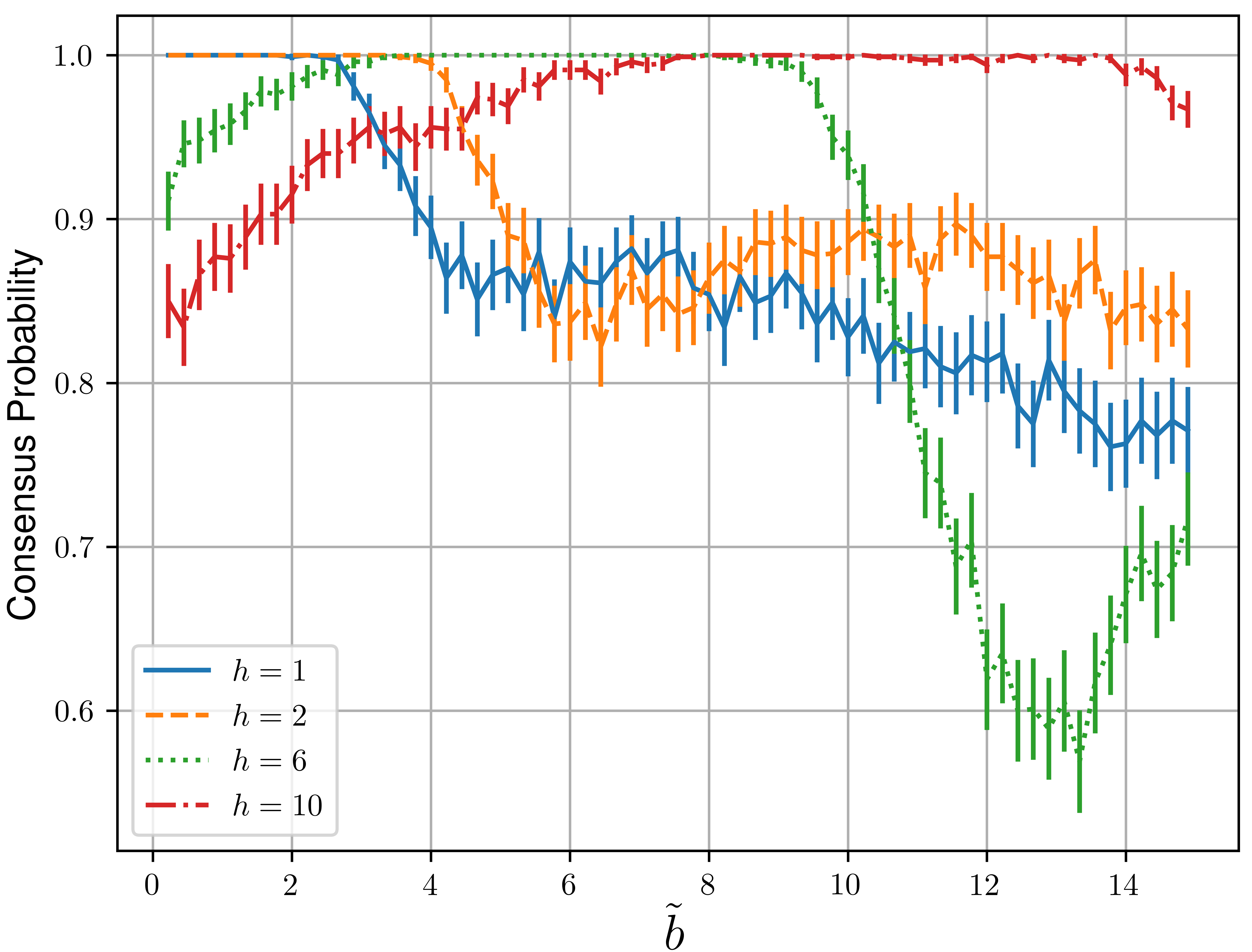}
\caption{Two Block Model}
\label{fig6a}
\end{subfigure}
\hfill
\begin{subfigure}[c]{0.49\textwidth}
\includegraphics[scale=0.415]{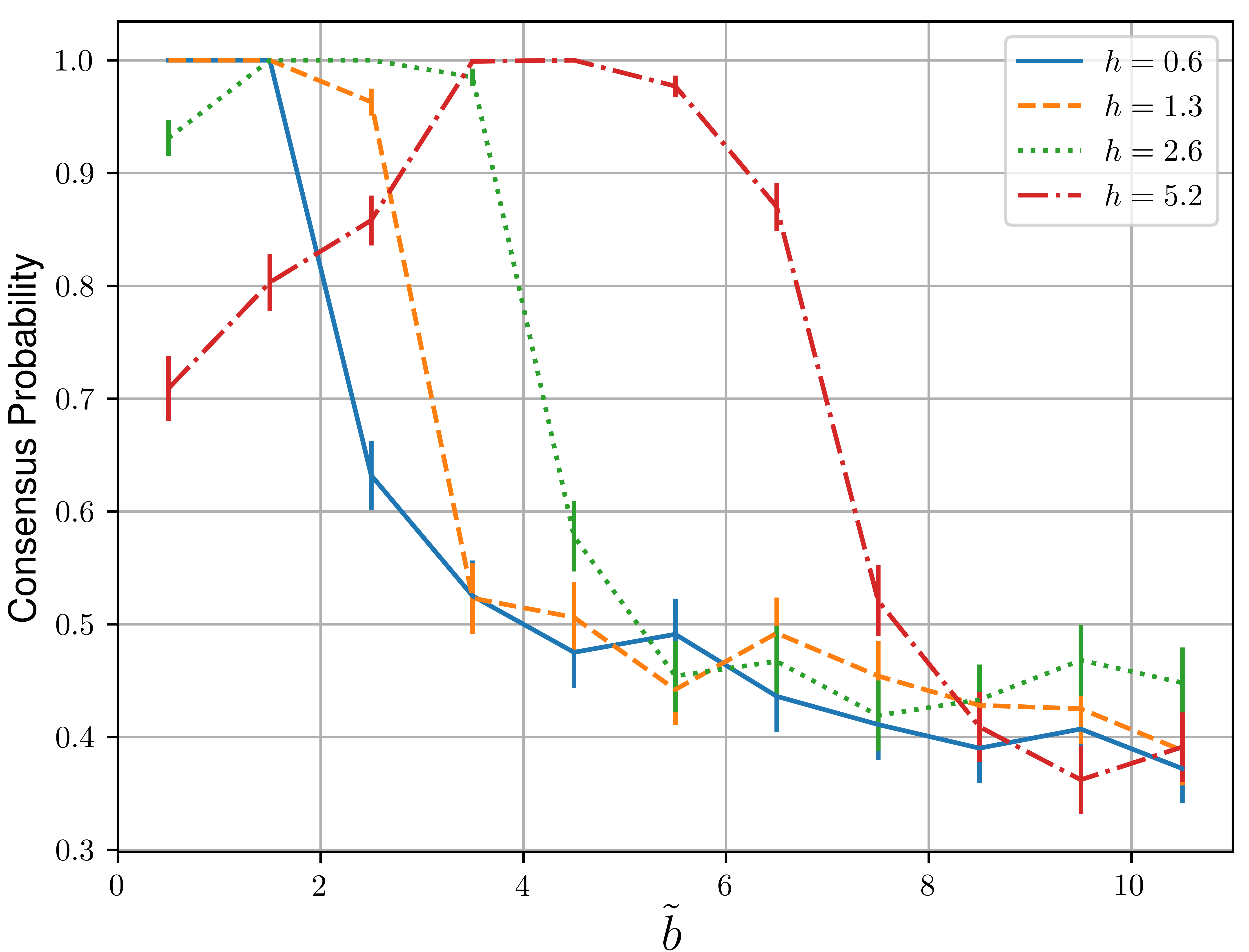}
\caption{Random Graphs}
\label{fig6b}
\end{subfigure}\\
\begin{subfigure}[c]{\textwidth}
\centering
\includegraphics[scale=0.415]{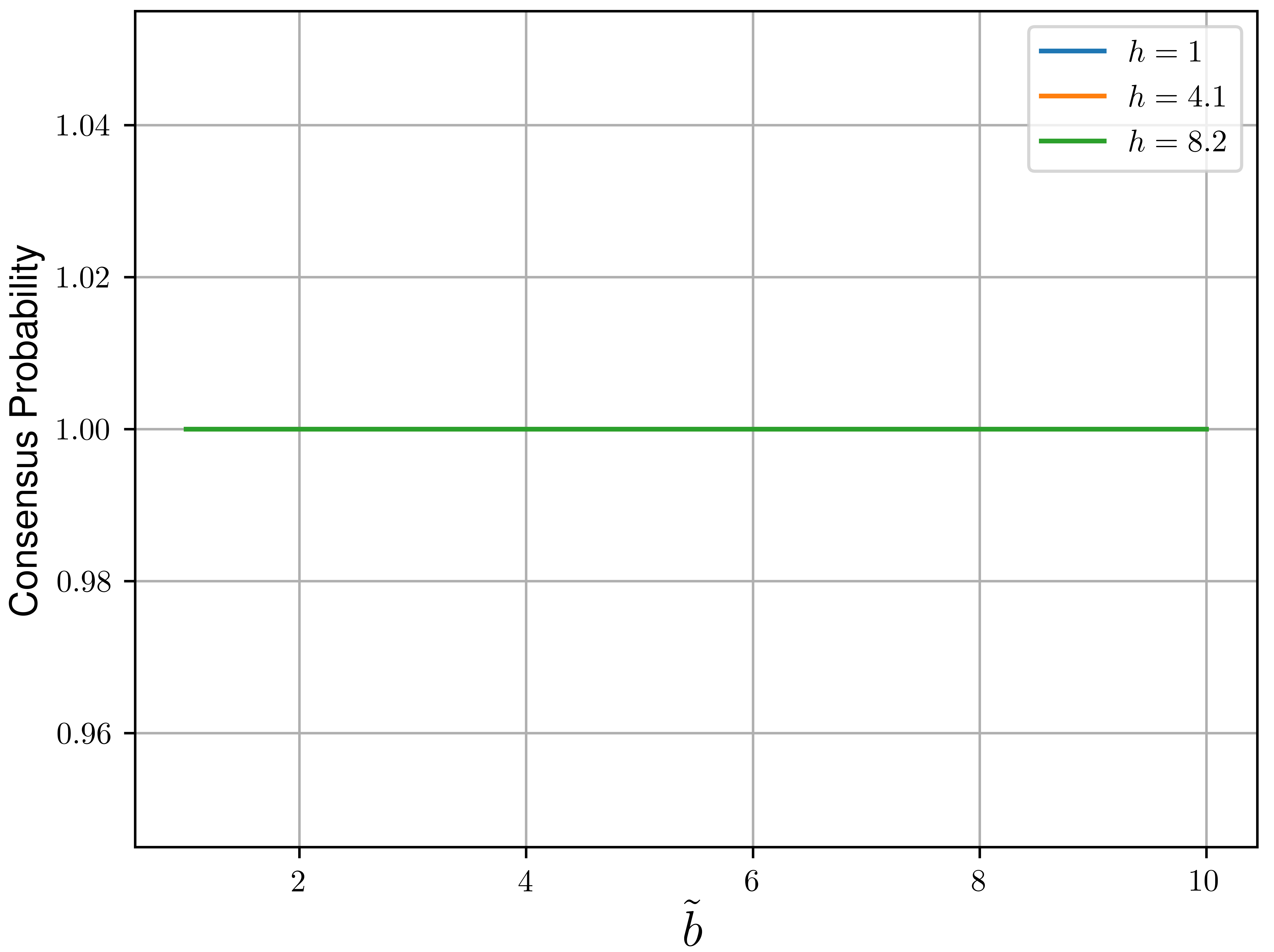}
\caption{\texttt{ego-facebook}}
\label{fig6c}
\end{subfigure}
 \caption{In Figure~\ref{fig6a} we show the results in the same setting as Figure~\ref{fig2a}.\\
In Figure~\ref{fig6b} we show how the consensus probability changes when $\delta = 0.125$, $\wout = 1, \win \in \{0.5,1,2,4\}$ and $b$ is set to have $\tb \in \{0.5,1.5,\ldots,10.5\}$. As Figure\ref{fig2}, $h$ is the expected value of the homophily ratio over all runs involving the same value of $\win$.\\
In Figure~\ref{fig6c} we show the results in the same setting as Figure~\ref{fig2b}.\\
Note that in this experiment the initial opinions are not constrained to diverge; indeed, they are still drawn uniformly at random in intervals $[l_L,h_L]$ and $[l_R, h_R]$, but both intervals are equal to $[-1,1]$.\\
The confidence intervals are shown as error bars. Similar experiments have been run also on the remaining network models with very similar results.
}
\label{fig6}
\end{figure}
However, differently from what happens for the case of divergent initial opinions, we here highlight a difference between experimental and theoretical results. The latter ones show that the thresholds among the two phases should depend on the homophily ratio $h$. However, this dependence does not appear in experiments. We leave open the problem of investigating about the reasons behind this discrepancy.
% ({\color{red}{dall'analisi teorica vediamo che inizialmente dovremmo avere una consens probability bassa fino a $\tau_4$ quando $h$ è sufficientemente alto da rendere $\tau_4$ diverso da zero. Più $\delta$ è basso più questo valore di $h$ dovrebbe essere alto. Durante gli esperimenti invece abbiamo osservato che, per qualunque tipo di rete, il valore di $h$ tale per cui ho che per valori bassi di $\tb$ probabilità di consenso vicino a zero non dipende da $\delta$.}})

\section{Asynchronous Updates}
\label{seAsync}
In this section we will focus on opinion dynamics with asynchronous updates, where at each time step, a single agent, arbitrarily chosen, is allowed to update his/her opinion.

We will show that in this setting, the game admits a generalized potential function \citep{monderer}.
%, hence the game has the Finite Improvement Property  
Consequently, unlike the synchronous case, the dynamics always converges to a stable state.

We will see that in the asynchronous case we may also establish useful bounds on the convergence of the dynamics.

Finally, we will focus on consensus, by providing experimental evidence that the results described in the previous section still hold in the asynchronous case.

\subsection{Convergence}
\label{convergence}
We start by showing that the proposed opinion game is a generalized ordinal potential game and thus the opinion dynamics with asynchronous updates always converges to a stable profile \citep{monderer}.

\begin{theorem}
\label{thm:stable}
For every $G=(V,E,w)$ and every $b\geq 0$, the opinion dynamics with asynchronous updates always converges to a stable opinion profile $\bfx$.
\end{theorem}
\begin{proof}
 The theorem follows by showing that the function $\Phi(\bfx) = \sum_{i \in V} b_i (x_i - s(x_i))^2+P(x)$ is a generalized ordinal potential function for the game described above, where $\bfx = (x_1, \ldots, x_n)$ is an opinion profile and $P(\bfx) = \sum_{(i,j) \in E} w_{i,j}(x_i-x_j)^2$.
Let $P_i(\bfx) = \sum_{j \colon (i,j) \in E} w_{i,j} (x_i-x_j)^2$, then the cost of agent $i$ given the opinion profile $\bfx$ is $c_i(\bfx) = b(x_i-s(x_i))^2 + P_i(x)$.

We call an edge $(e_1,e_2)$ a discording edge if its endpoints have different opinions, i.e., $x_{e_1} \neq x_{e_2}$. Since the graph $G$ is undirected, then $\sum_{i \in V} P_i(\bfx)$ is twice the sum of weights of all the discording edges with respect to the opinion profile $\bfx$. Hence, $P(\bfx) = \frac{1}{2}\sum_{i \in V} P_i(\bfx)$.

Denote by $(y_i,\bfx_{-i})$ the opinion profile obtained from $\bfx$ when player $i$ switches from opinion $x_i$ to opinion $y_i$. Then
\begin{equation*} 
c_i(\bfx) - c_i(y_i,\bfx_{-i}) = b_i\left[(x_i-s(x_i))^2-(y_i-s(x_i))^2\right]+P_i(\bfx)-P_i(y_i,\bfx_{-i}).
\end{equation*}
 From the definition of function $s$ and the choice of $\lambda = \frac{1}{2}$, it follows that $(y_i-s(x_i))^2 \geq (y_i - s(y_i))^2$. Hence:
\begin{equation}
 \label{eq:cost_diff}
  c_i(\bfx) - c_i(y_i,\bfx_{-i}) \leq b_i\left[(x_i-s(x_i))^2-(y_i-s(y_i))^2\right]+P_i(\bfx)-P_i(y_i,\bfx_{-i}).
\end{equation}

The difference in the function $\Phi$ between the same pair of profiles is:
\begin{align*}
\Phi(\bfx) - \Phi(y_i, \bfx_{-i}) & = \sum_{k \in V} b_k (x_k - s(x_k))^2+P(\bfx) - \sum_{\begin{subarray}{c}j \in V\\j \neq i\end{subarray}} b_j (x_j - s(x_j))^2\\
 & \qquad - b_i(y_i-s(y_i))^2 - P(y_i, \bfx_{-i})\\
 & = b_i\left[(x_i-s(x_i))^2-(y_i-s(y_i))^2\right]+P(\bfx)- P(y_i, \bfx_{-i})
\end{align*}
Let $D_i(\bfx) = \sum_{\begin{subarray}{c}(j,z) \in E\\j,z \neq i\end{subarray}} w_{j,z} (x_j - x_z)^2$ be the sum of the weights of discording edges not incident on $i$ in the opinion profile $\bfx$. It is immediate to see that $D_i(\bfx)$ is not affected by the deviation of player $i$, and thus $D_i(\bfx) = D_i(y_i, \bfx_{-i})$. Thus, we can write $P(\bfx)$ as:
\begin{align*}
P(\bfx) & = \sum_{\begin{subarray}{c}(j,z) \in E\\j,z \neq i\end{subarray}} w_{j,z} (x_j - x_z)^2 + \sum_{j \colon (i,j) \in E} w_{i,j} (x_i-x_j)^2 = D_i(\bfx)+P_i(\bfx)
\end{align*}
Similarly, we have $P(y_i, \bfx_{-i}) = D_i(y_i,\bfx_{-i})+P_i(y_i, \bfx_{-i}) = D_i(\bfx) + P_i(y_i, \bfx_{-i})$. Consequently we have $P(\bfx)-P(y_i,\bfx_{-i}) = P_i(\bfx)-P_i(y_i,\bfx_{-i})$, and thus
 \begin{equation}
\label{eq:phy_diff}
 \Phi(x) - \Phi(x_{-i}, y_i)  = b_i\left[(x_i-s(x_i))^2-(y_i-s(y_i))^2\right]+P_i(x)- P_i(x_{-i},y_i).
\end{equation}
By \eqref{eq:cost_diff} and \eqref{eq:phy_diff} it follows that if $c_i(\bfx) - c_i(y_i, \bfx_{-i}) > 0$, then $\Phi(\bfx) - \Phi(y_i, \bfx_{-i}) > 0$, and thus $\Phi$ is a generalized ordinal potential function, as desired.
 \end{proof}
%In our analysis we assume that, if a player does not change its opinion it cannot be selected until a different player changes its opinion. Trivially, if no player changes its opinion a Nash Equilibrium is reached. Basically, this assumption guarantee that it is not possible to always choose the same player and it is reasonable since it makes no sense to allow a player to change its opinion if it did not change its opinion in the same conditions.\\
Observe that from an initial opinion profile $x^0 = \{x_0^0, \ldots, x_n^0\}$, the opinion dynamics with asynchronous updates may converge to different stable states, depending on the order in which agents are chosen for updating their opinions. For example, let $G$ be a three-player social network, let
$V = \{0, 1, 2\}$, $E = \{(0,1),(1,2),(0,2)\}$, $x^0 = \{-1, 0.5, 0\}$, $\Theta = \{-1, -0.5, 0, 0.5, 1\}$, $b = 1$, and $\lambda = 0.5$. Recall that the opinion of player $i$ at step $t$ is
$$x_i^t = \argmin_{x}{\{b(x-s(x_i^{t-1}))^2 + \sum_{j:(i,j) \in E}{w_{i,j}(x-x_j^{t-1})^2}\}}.$$
Thus, if players' selection order is $\{\text{Player }0$, $\text{Player } 2$, $\text{Player } 1$, $\text{Player } 0$, $\text{Player } 1$, $\text{Player } 2\}$, then the dynamics reach the stable opinion profile $x^{eq} = \{-0.5, -0.5, -0.5\}$; if, instead, the selection order is $\{\text{Player } 2$, $\text{Player } 1$, $\text{Player } 0$, $\text{Player } 0, \text{Player } 1$, $\text{Player } 2\}$, then the stable opinion profile is $x^{eq} = \{-1, -1, -1\}$.
% 
% In section~\ref{async_res} we will discuss how the experimental results obtained for the asynchronous dynamics are very similar to the ones obtained for the synchronous case, showing a large adherence with the findings discussed in section~\ref{sync}.
% In the next section we will determine the bound for the convergence of the \emph{asynchronous} dynamics. 

We are also able to bound the time that the opinion dynamics takes to converge to a stable state.
% In Theorem~\ref{thm:stable} we prove that the opinion game admits a generalized ordinal potential function, hence, the opinion dynamics with asynchronous update always converges to a stable state. Moreover
Specifically, we observe that convergence time can in general be exponential in the number of agents, as stated by the next theorem.
\begin{theorem}
\label{exp_step}
There is a social network $G=(V, E)$ with $\lvert V\rvert = n$ and an opinion set $\Theta$ such that the corresponding opinion dynamics with asynchronous updates takes an number of steps to converge to a stable state that is exponential in $n$.
\end{theorem}

Nevertheless, through the analysis of the generalized potential function defined in Theorem~\ref{thm:stable}, we can determine polynomial upper bounds on the number of steps needed to converge to a stable state, whenever the weights of edges and the social media influence have bounded precision $k$, i.e. they can be represented with at most $k$ digits after the decimal point.
\begin{theorem}
\label{conv_bound}
Given a social network $G = (V,E)$ with $|V| = n$ and an opinion set $\Theta$ with discretization factor $\delta$, if both the strenght $b$ of social media influence and the weights $(w_e)_{e \in E}$ of the edges have bounded precision $k$, then the opinion dynamics with asynchronous updates converges to a Nash Equilibrium in $O(10^k4b\frac{n}{\delta^2}+10^k4w_{max}\frac{n^2}{\delta^2})$, where $w_{\max} = \max_{e \in 	E} w_e$.
\end{theorem}

The proofs of Theorem~\ref{exp_step} and Theorem~\ref{conv_bound} resemble the ones used for proving similar results by \cite{goldberg}. Anyway, we include them in Appendix~\ref{apx:conv_proofs} for sake of completeness.

\subsection{Impact on Consensus}
\label{async_res}
We will now focus on the behavior of the opinion dynamics with asynchronous updates with respect to convergence to consensus. We will give experimental evidence that, even in this case, the behavior of the dynamics resembles the one observed in the experiments shown in section~\ref{sync_experiments}, and hence it reflects the theoretical findings obtained through the analysis of the synchronous case (section~\ref{sync}). Specifically, we performed simulations in the very same settings described in section~\ref{sync_experiments}. Therefore, we run experiments on stochastic two-block model graphs, random graphs and on real graphs. As observed in the synchronous case, we observe that the trend of the consensus probability does not change by varying the model's parameters. We omit these results for the seek of readability, since they are very similar to the ones showed in section~\ref{sync_experiments}. Moreover, we compute weights, homophily ratio, social media influence, consensus probability and confidence interval exactly as discussed in section~\ref{sync_experiments}. Unlike the synchronous case, we have to define how, at each time step, we determine the agent that can update her opinion. Specifically, let $\mathcal{A}_t$ the set of agents at step $t$ for which the best-response is to change their current opinion, the update policy is: at each round $t$ we sample an agent from $\mathcal{A}_t$ uniformly at random.

From Figure~\ref{async_fig1} we observe that the probability of reaching a consensus is the same when $\tb$ and $h$ are unchanged, even if we change the value of $b, \aout, \ain$, exactly as in the case of synchronous case.

\begin{figure}[htb]
% \subcaptionbox{Two Block Model -- ToUpdate}[0.5\linewidth][c]{\includegraphics[scale=0.41]{/ConsensusProbabilityCompAlpha/BlockModel}}
\begin{subfigure}[c]{0.49\linewidth}
% \centering
\includegraphics[scale=0.4]{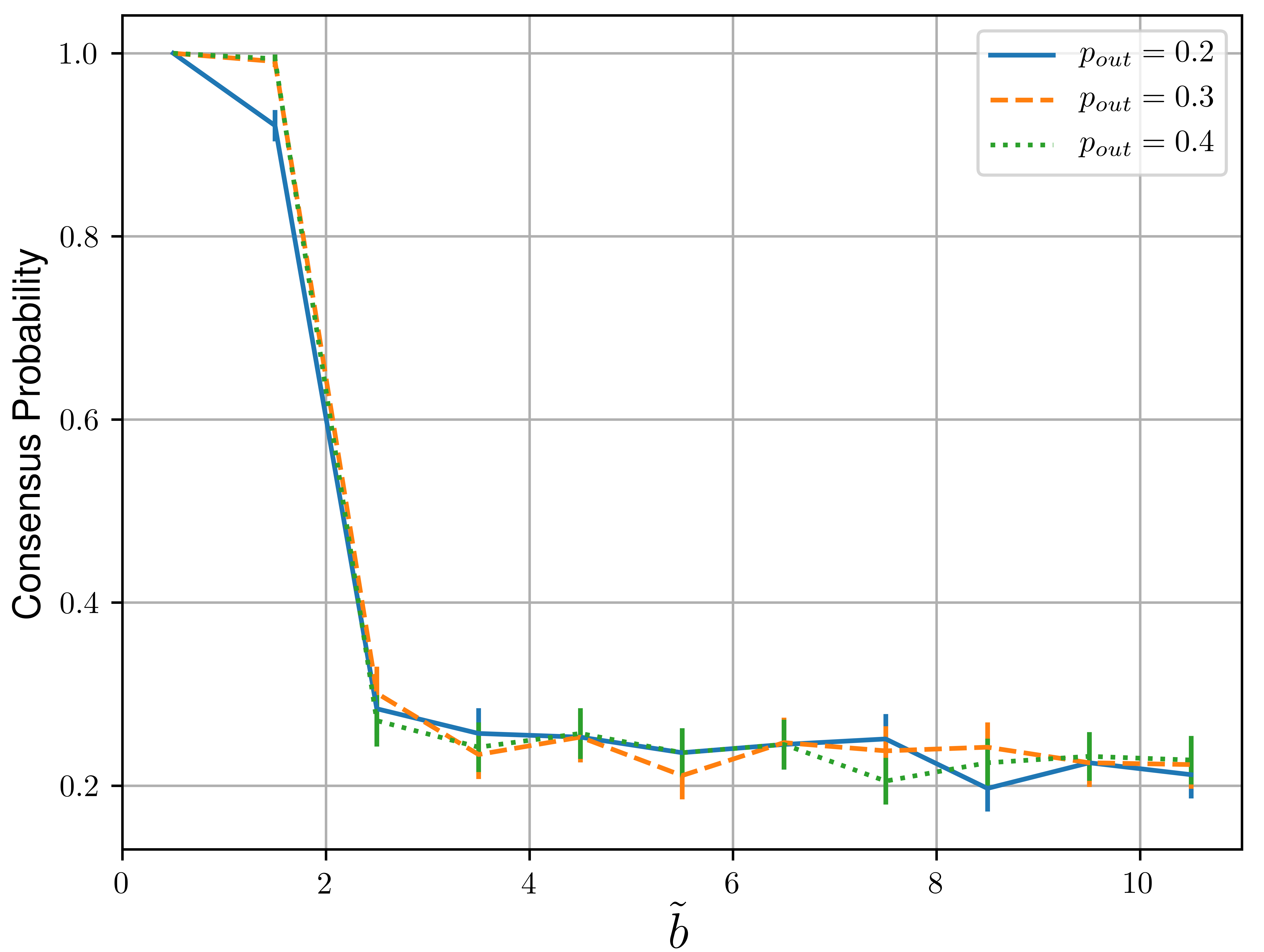}
\caption{Two Block Model}
\label{figAsync1a}
\end{subfigure}%
\begin{subfigure}[c]{0.49\linewidth}
% \centering
\includegraphics[scale=0.4]{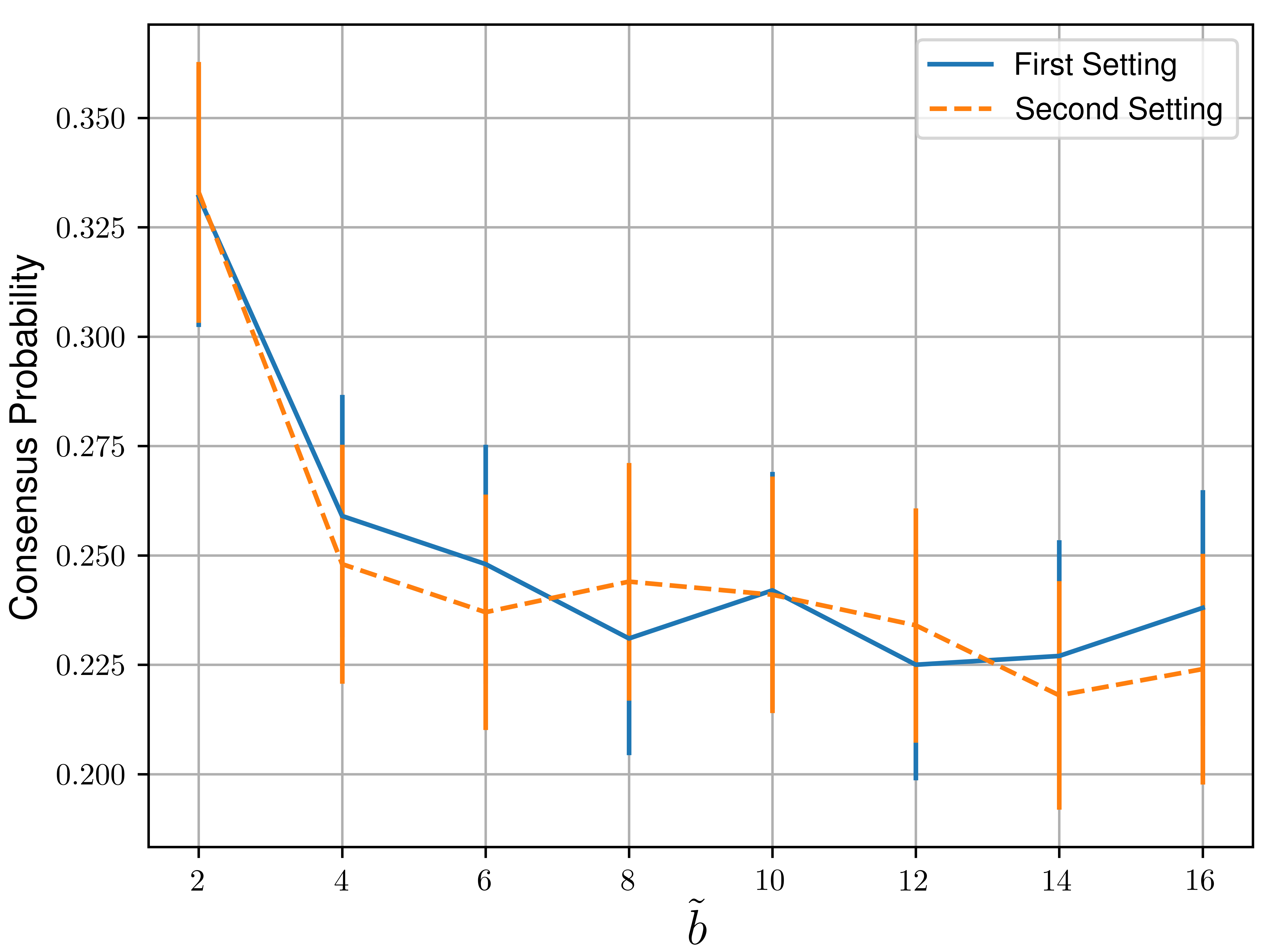}
\caption{\texttt{ego-Facebook}}
\label{figAsync1b}
\end{subfigure}
 \caption{Figures show the trend of consensus probability in the same setting as Figure~\ref{fig1}, except that the update of opinions is asynchronous.\\
The confidence intervals are shown as error bars. Similar results have been obtained for the other network models.}
\label{async_fig1}
\end{figure}
The results showed in Figure~\ref{async_fig2} highlight that the findings obtained in the section~\ref{sync} on the dependence between the relative amount of social media influence $\tb$, the homophily ratio $h$, and the consensus probability hold also for the asynchronous case.
\begin{figure}[htb]
 \begin{subfigure}{0.49\textwidth}
\centering
\includegraphics[scale=0.415]{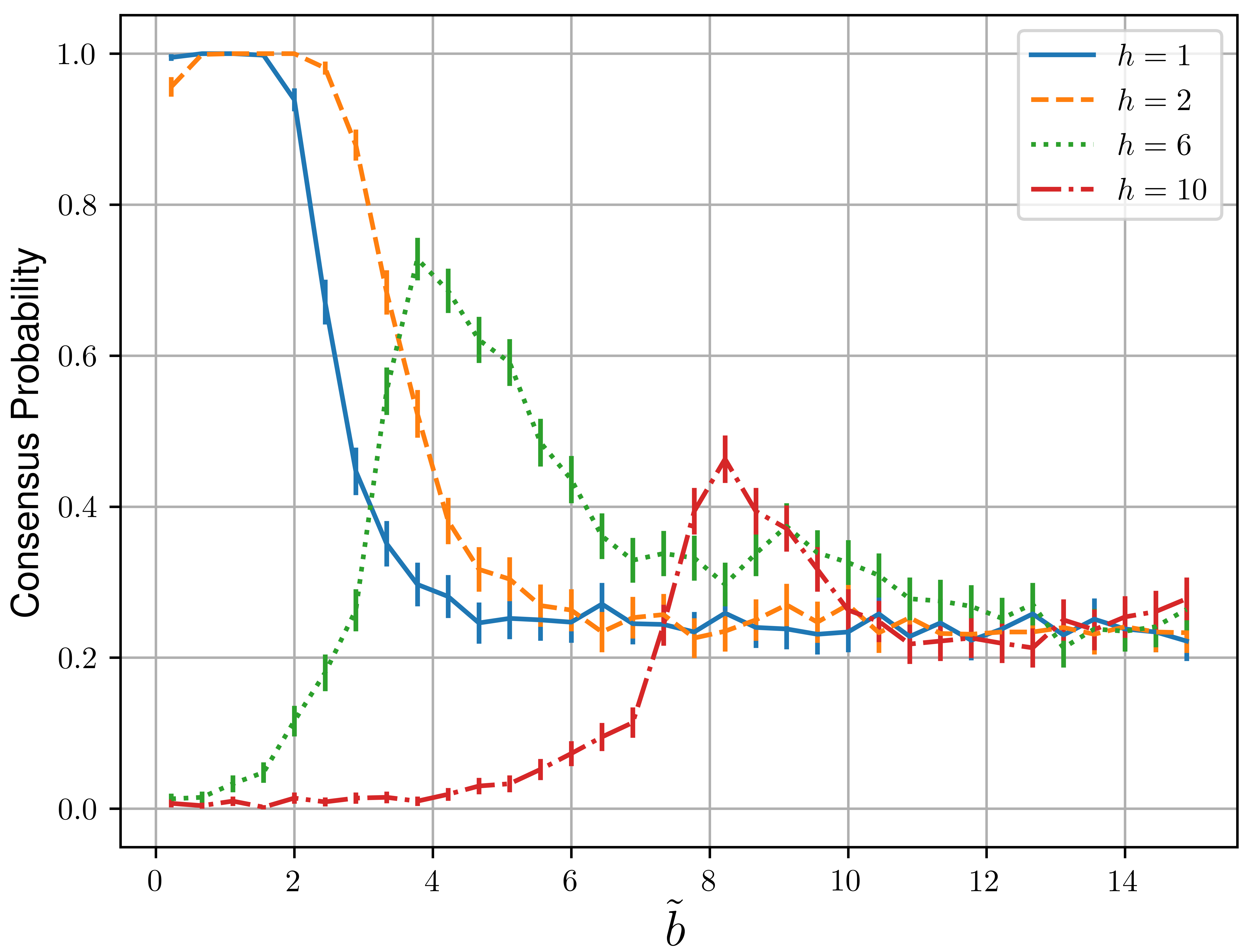}
\caption{Two Block Model}
\label{async_fig2a}
\end{subfigure}
\hfill
\begin{subfigure}{0.49\textwidth}
\centering
\includegraphics[scale=0.415]{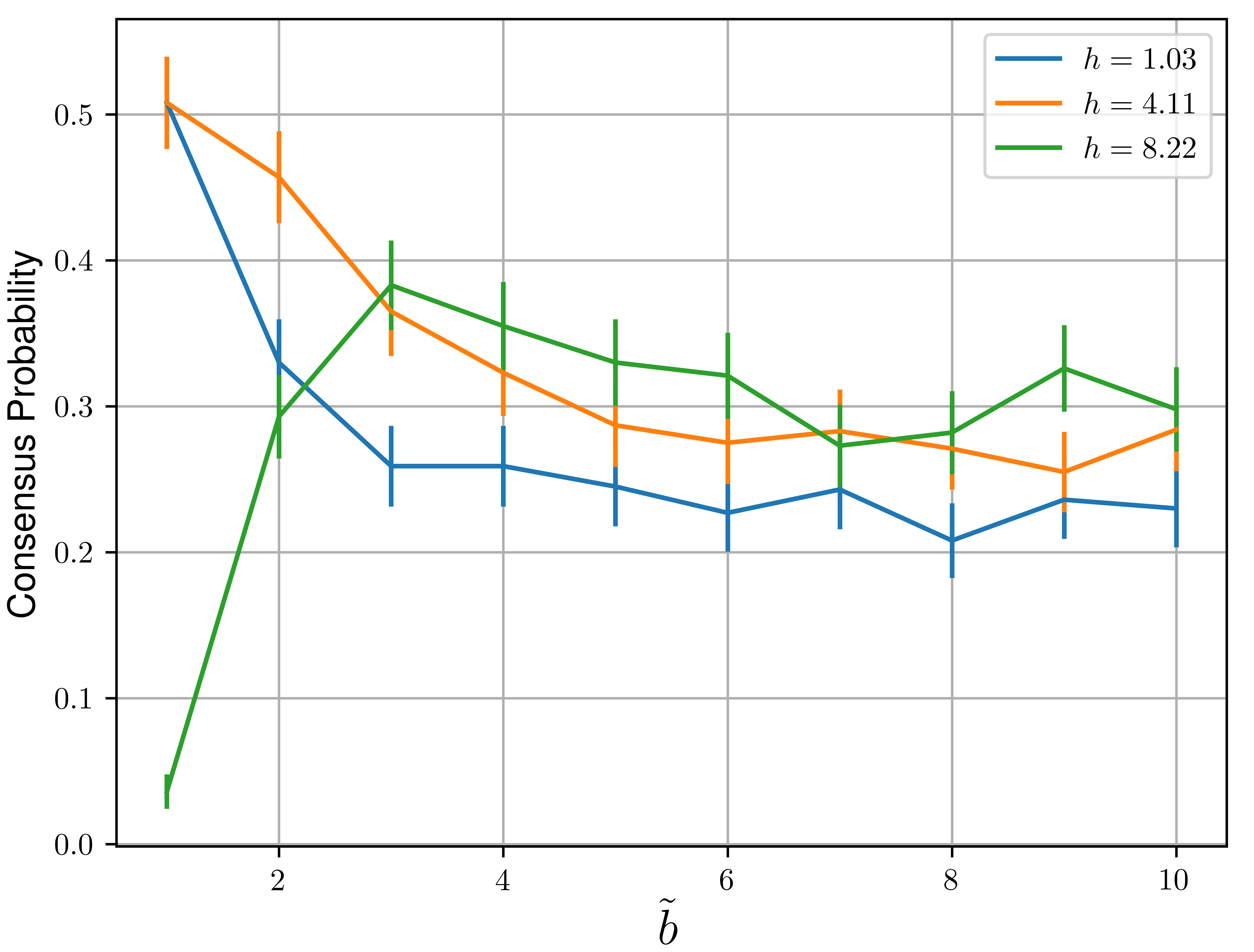}
\caption{\texttt{ego-Facebook}}
\label{async_fig2b}
\end{subfigure}
 \caption{Figures show the trend of consensus probability in the same setting as Figure~\ref{fig2}, except that the update of opinions is asynchronous.\\
The confidence intervals are shown as error bars. Similar results have been obtained for the other network models.}
%  and are omitted due to space limits.
\label{async_fig2}
\end{figure}
As discussed in section~\ref{sync_experiments}, we observe that convergence to consensus is affected by the initial opinions of agents. Specifically, we observe that the consensus probability decreases as the average opinion of agents in a component goes to 1, with a sharp change of probability around $\lambda = 0.5$. We run experiments in which we impose initial opinions to be larger or smaller than $\lambda$. We observe that the decrement of the consensus probability occurs as soon as the value of $\tb$ is around the thresholds $\tau_1$ and $\tau_2$, as indicated by findings obtained in section~\ref{sync}. The results are very similar to the ones showed in Figure~\ref{fig3}, Figure~\ref{fig4} and Figure~\ref{fig5}, hence for the sake of the brevity we omit them.

Moreover, we consider the case in which initial opinions do not diverge. The results confirms that the smaller is the homophily ratio the smaller $\tb$ need to be to make the probability of consensus large by allowing consensus on an opinion different from $-1, 0, 1$, as proved in Theorem~\ref{thm:conv_small}. As observed in section~\ref{sync_experiments} for real dataset this behaviour is not confirmed. Further, the results highlight a difference between experimental and theoretical results. Specifically, the latter ones show that the threshold should depend on the homphily ratio $h$, but this dependence does not appear in experiments. We leave open the problem of investigating about the reasons behind these discrepancies. The results obtained are very similar to the ones shown in Figure~\ref{fig6}, hence we omit them.

\section{Conclusions}
In this work we analyzed the impact of social media recommendations on opinion formation processes, when opinions may assume only discrete values, as is the case of several electoral settings. We focused mainly on how and how much the social media may influence the likelihood that agents reach a consensus.
% that is an important benchmarking largely adopted in previous theoretical literature and with huge practical applications. 
Clearly, it would be interesting also to deepen our analysis by evaluating how the social media can influence, not only the probability of consensus, but also the kind of equilibria that can be reached by the opinion formation process.

In this work we focused on a classical opinion formation model.
% : this choice has been motivated, on one side, by the fact that this model are much more settled in literature, and on the other side, on the will of untying the social media impact from the complexities of the model.
However, we believe that it would be undoubtedly interesting to analyze whether our results extend to more complex (but more realistic) opinion formation models.

In our analysis, we restrict the opinion space of the social media $\Omega$ to only three values. To consider an higher cardinality of $\Omega$ would be clearly of interest, but we will expect that such an analysis will give results very similar in spirit to the one proved in our work (but with an explosion of possible cases). Similar considerations can be done about extending our mono-dimensional representation of opinions to higher dimensional representations.

Even if, our experimental results highlight a large adherence to the theoretical findings obtained for the symmetric two-block model, some small differences exist among the results for different network structures (mainly, in the case that initial opinions are convergent). It would be then interesting to understand whether and how these difference may be motivated through a detailed study of the relationship among the impact of the social influence and the structural and topological properties of the social network.

\pagebreak

% \begin{appendices}
\appendix

\section{Missing Proofs from Section~\ref{preliminary}}
\label{apx:pre_res}
\begin{proof}[Proof of Lemma~\ref{lem:close_zero}]
We consider only the case that $x_i^t < -\lambda$. The case for $x_i^t > \lambda$ is symmetric and hence omitted.

Recall that, by Lemma~\ref{lem:TwoPlayerModel}, $x_i^t = x_P^t$, where $P$ is the block at which $i$ belongs, and every $j \notin P$ has $x_j^t = x_{\notP}^t$.
Since $x_i^t < -\lambda$, then $s(x_P^t) = 1$. Thus, by Lemma \ref{lem:boundOnAvg}, $x_i^{t+1}\geq 0$ only if $\frac{-b+\aout x_{\notP}^t+\ain x_P^t}{b+\ain+\aout} \geq \frac{\delta}{2}$.
By dividing both sides by $\aout$ and recalling that $\frac{b}{\aout} = \tb$ and $\frac{\ain}{\aout} = h$, we have that $x_i^{t+1}\geq 0$ only if 
\begin{equation}
\label{eq:condlem6}
 \tb \leq \frac{x_{\notP}^t-\frac{\delta}{2}+h(x_P^t-\frac{\delta}{2})}{\frac{\delta}{2}+1}.
\end{equation}
It is immediate to check that the r.h.s. of \eqref{eq:condlem6} is maximized by taking $x_P^t = -\lambda-\delta$, and $x_{\notP}^t = 1$. By substituting these values in \eqref{eq:condlem6}, we 
achieve that 
$x_i^{t+1}$ can be greater than $0$ only if $\tb \leq \frac{2-\delta-(2\lambda+3\delta)h}{2+\delta}$, as desired.
% \item By Lemma \ref{lem:boundOnAvg}, if $x_i^t > \lambda$, then the next step opinion $x_i^{t+1}$ will be lower than zero only if:
% \begin{equation*}
% \frac{b+\aout x_{\notP}^t+\ain x_P^t}{b+\ain+\aout} \leq - \frac{\delta}{2}
% \end{equation*}
% Grouping by $\frac{b}{\aout} = \tb$, we have the following inequality:
% \begin{equation*}
% \tb \leq -\frac{x_{\notP}^t+\frac{\delta}{2}+h(x_P^t+\frac{\delta}{2})}{\frac{\delta}{2}+1}
% \end{equation*}
% We can easily compute $\max_{(x_P^t, x_{\notP}^t)}{\left( -\frac{x_{\notP}^t+\frac{\delta}{2}+h(x_P^t+\frac{\delta}{2})}{\frac{\delta}{2}+1} \right)} =_{(x_P^t = \lambda+\delta, x_{\notP}^t = -1)} \frac{2-\delta-h(2\lambda+3\delta)}{2+\delta}$. Consequently, if $\tb > \frac{2-\delta-h(2\lambda+3\delta)}{2+\delta}$, then the next step opinion can not be lower than zero;\\
\end{proof}

\begin{proof}[Proof of Lemma~\ref{lem:zero}]
We consider only the case that $x_i^t \in [-\lambda, 0)$. The case for $x_i^t \in (0, \lambda]$ is symmetric and hence omitted.

Recall that, by Lemma~\ref{lem:TwoPlayerModel}, $x_i^t = x_P^t$, where $P$ is the block at which $i$ belongs, and every $j \notin P$ has $x_j^t = x_{\notP}^t$.
Since $x_i^t \in [-\lambda, 0)$, then $s(x_P^t) = 0$. Thus, by Lemma \ref{lem:boundOnAvg}, $x_i^{t+1}\geq 0$ only if $\frac{\aout x_{\notP}^t+\ain x_P^t}{b+\ain+\aout} \geq -\frac{\delta}{2}$.
By dividing both sides by $\aout$ and recalling that $\frac{b}{\aout} = \tb$ and $\frac{\ain}{\aout} = h$, we have that $x_i^{t+1}\geq 0$ only if 
\begin{equation}
\label{eq:condlem7}
 \tb \geq -\frac{2}{\delta}\left( (x_{\notP}^t+\frac{\delta}{2})+h(x_P^t+\frac{\delta}{2})\right).
\end{equation}
It is immediate to check that the r.h.s. of \eqref{eq:condlem7} is minimized by choosing $x_P^t = -\delta$, and for $x_{\notP}^t = 1$. By substituting these values in \eqref{eq:condlem7}, we achieve that 
$x_i^{t+1}\geq 0$ only if $\tb \geq h - \left(\frac{2}{\delta} + 1\right)$, as desired.
% \item if $x_i^t \in [\delta, \lambda],$ then $s(x_i^t) = 0$. By Lemma \ref{lem:boundOnAvg}, the next step opinion $x_i^{t+1}$ will be lower than or equal to zero only if:
% \begin{equation*}
% \frac{\aout x_{\notP}^t+\ain x_P^t}{b+\ain+\aout} \leq \frac{\delta}{2}
% \end{equation*}
% Grouping by $\frac{b}{\aout} = \tb$, we have the following inequality:
% \begin{equation*}
% \tb \geq \frac{2}{\delta}\left( (x_{notP}^t-\frac{\delta}{2})+h(x_P^t-\frac{\delta}{2})\right)
% \end{equation*}
% Therefore, if $\tb < \min_{(x_P^t, x_{\notP}^t)}{\left(\frac{2}{\delta}\left( (x_{notP}^t-\frac{\delta}{2})+h(x_P^t-\frac{\delta}{2})\right) \right)} =_{(x_P = \delta, x_{\notP} = -1)} h-(\frac{2}{\delta}+1)$, then the next step opinion can not be lower than or equal to zero.
\end{proof}

\begin{proof}[Proof of Lemma~\ref{lem:far_zero}]
We consider only the case that $x_i^t < -\lambda$. The case for $x_i^t > \lambda$ is symmetric and hence omitted.

Recall that, by Lemma~\ref{lem:TwoPlayerModel}, $x_i^t = x_P^t$, where $P$ is the block at which $i$ belongs, and every $j \notin P$ has $x_j^t = x_{\notP}^t$.
Since $x_i^t < -\lambda$, then $s(x_P^t) = 1$. Thus, by Lemma \ref{lem:boundOnAvg}, $x_i^{t+1}\geq 0$ only if $\frac{-b+\aout x_{\notP}^t+\ain x_P^t}{b+\ain+\aout} \geq -\frac{\delta}{2}$.
By dividing both sides by $\aout$ and recalling that $\frac{b}{\aout} = \tb$ and $\frac{\ain}{\aout} = h$, we have that $x_i^{t+1}\geq 0$ only if 
\begin{equation}
\label{eq:condlem8}
 \tb \leq \frac{x_{\notP}^t+\frac{\delta}{2}+h(x_P^t+\frac{\delta}{2})}{1-\frac{\delta}{2}}.
\end{equation}
It is immediate to check that the r.h.s. of \eqref{eq:condlem8} is minimized by taking $x_P^t = -\lambda-\delta$, and $x_{\notP}^t = 1$. By substituting these values in \eqref{eq:condlem8}, we achieve that 
$x_i^{t+1}$ can be at least $0$ only if $\tb \leq \frac{2+\delta-(2\lambda+\delta)h}{2-\delta}$, as desired.
% \item if $x_i^t > \lambda$, then $s(x_i^t) = 1$. By Lemma \ref{lem:boundOnAvg}, the next step opinion $x_i^{t+1}$ can be lower than or equal to zero only if:
% \begin{equation*}
% \frac{b+\aout x_{\notP}^t+\ain x_P^t}{b+\ain+\aout} \leq \frac{\delta}{2}
% \end{equation*}
% Grouping by $\frac{b}{\aout} = \tb$, we have the following inequality:
% \begin{equation*}
% \tb \leq \frac{-x_{\notP}^t+\frac{\delta}{2}-h(x_P^t-\frac{\delta}{2})}{1-\frac{\delta}{2}}
% \end{equation*}
% Therefore, if $\tb > \max_{(x_P^t, x_{\notP}^t)}{\left(\frac{-x_{\notP}^t+\frac{\delta}{2}-h(x_P^t-\frac{\delta}{2})}{1-\frac{\delta}{2}}\right)} =_{(x_P = \lambda+\delta, x_{\notP} = -1)} \frac{2+\delta-(2\lambda+\delta)h}{2-\delta}$, then the next step opinion can not be lower than or equal to zero.
\end{proof}

\section{Missing Proofs from Section~\ref{sync_results}}
\label{apx:sec3}
\begin{proof}[Proof of Theorem~\ref{thm:div_large_one}]
 Suppose w.l.o.g. that $x_L^0 < -\lambda $, and $x_R^0 \geq 0$.
 Suppose first that $\tb > \tau^*(h)$. Since $\tau^*(h) \geq \tau_1(h)$, then, by Lemma~\ref{lem:far}, no agent $i \in L$ can take an opinion $x_L^1 \geq -\lambda$. Moreover, since $\tau^*(h) \geq \tau_2(h)$, then, by Lemma~\ref{lem:close}, no agent $j \in R$ can take an opinion $x_R^1 < 0$. Hence, after the first time step we still have $x_L^1 < -\lambda$ and $x_R^1 \geq 0$. Then, we can iteratively apply the same argument above to conclude that the opinions of the two blocks never converge to a consensus profile.
 
 As for the remaining cases, the argument is exactly the same as discussed in the proof of Theorem~\ref{thm:div_large_both}.
\end{proof}

\begin{proof}[Proof of Theorem~\ref{thm:div_small}]
Suppose that $\tb > \max\{\tau_2(h), \tau_3(h)\}$. W.l.o.g., suppose that $x_L^0 \leq 0 \text{ and } x_R^0 \geq 0$. Then by Lemma~\ref{lem:close} and Lemma~\ref{lem:close_zero}, no agent $i \in L$ can take an opinion $x_L^1 > 0$  and no agent $j \in R$ can take an opinion such that $x_R^1 < 0$. Hence, after the first step $x_L^1 \leq 0 \text{ and } x_R^1 \geq 0$. Then, we can iteratively apply the same argument to conclude that only consensus on $0$ is possible.

For the case where $0 < \tb \leq \max{0,\tau_4(h)}$, the argument is exactly the same as discussed in the proof of Theorem~\ref{thm:div_large_both}.
\end{proof}

\section{Missing Proofs from Section~\ref{convergence}}
\label{apx:conv_proofs}
\subsection{Proof of Theorem~\ref{exp_step}}
\begin{proof}[Proof of Theorem~\ref{exp_step}]
In the following construction we assume $\delta = 0.5$.

A \emph{6-gadget $G$} is a graph $(V,E)$, with $\vert V \vert = 6$, $V = \{A, B, C, D, E, F\}$, with edges $(A,B)$, $(B,C)$ and $(C,D)$ having weights ${\epsilon, 2\epsilon, 3\epsilon}$ respectively, and edges $(D.E)$, $(B,F)$ and $(D,F)$ all weighting $4\epsilon$, for some $\epsilon > 0$. The value of the social media's influence is $b = \epsilon' < \epsilon$.

Let $A_0$ a further player with edges $(A_0, B)$ and $(A_0, D)$ of weight $4\epsilon$. $A_0$ allows $G$ to switch between vectors $(0, 0, 0, 0, 0, 1/2)$ and $(0, 1/2, 0, 1/2, 0, 1/2)$.
If the players $\{A, B, C, D, E, F\}$ have opinions $(0, 0, 0, 0, 0, 1/2)$ and $A_0$ is set to $1$, then we can have the following best-response sequence, that will be referred as \emph{switch-on cycle}:
\begin{equation*}
(0, 0, 0, 0, 0, 1/2) \rightarrow (0, 1/2, 0, 0, 0, 1/2) \rightarrow (0, 1/2, 0, 1/2, 0, 1/2)
\end{equation*}
On the other hand, if the players $\{A, B, C, D, E, F\}$  have opinions $(0, 1/2, 0, 1/2, 0, 1/2)$ and $A_0$ is set to $0$, then we can have the following best-response sequence that will be referred as \emph{switch-off cycle}:
\begin{multline*}
(0, 1/2, 0, 1/2, 0, 1/2) \rightarrow (1/2, 1/2, 0, 1/2, 0, 1/2) \rightarrow (1/2, 0, 0, 1/2, 0, 1/2) \\ \rightarrow (0, 0, 0, 1/2, 0, 1/2) \rightarrow (0,0,1/2,1/2,0,1/2) \rightarrow (0,1/2,1/2,1/2,0,1/2) \\ \rightarrow (1/2,1/2,1/2,1/2,0,1/2) \rightarrow (1/2,1/2,1/2,0,0,1/2) \rightarrow (1/2,1/2,0,0,0,1/2) \\ \rightarrow (1/2,0,0,0,0,1/2) \rightarrow (0,0,0,0,0,1/2)
\end{multline*}
During the switch-on cycle the opinion of $A$ does not change, while it follows the sequence $0 \rightarrow 1/2 \rightarrow 0 \rightarrow 1/2 \rightarrow 0$ in the switch-off cycle.

We now define an opinion game characterized by an exponentially large difference between the largest and the smallest edge's weight.
Consider an $n$ 6-gadgets $G_i$ with players $\{A_i, B_i, C_i, D_i, E_i, F_i\}$, with $i = 1, \ldots, n$. The edge' weights are parametrized by $\epsilon_i$, with $\epsilon_i < \epsilon_{i-1}, \forall i=1, \ldots, n$. Let the social media's influence $b < \epsilon_n$.
For each $i$ we connect $G_i$ with $G_{i-1}$ by having $A_{i-1}$ acting as a switch for $G_i$. We finally add the switch player $A_0$ for $G_1$. Consequently, the total number of players is $6n+1$.
We set the edges' weight to make the behaviour of $A_i$ not influenced by the edges $(A_i, B_{i+1})$ and $(A_i, D_{i+1})$ so that the opinion of $A_i$ always follows the opinion of $B_i$. Basically, we want that, whatever is the opinion of $B_{i+1} \text{ and } D_{i+1}$, the best-response for player $A_i$ always is the opinion of $B_i$. Specifically, the needed condition is $\epsilon_i > 8\epsilon_{i+1}+b$. Hence, it is sufficient to set $\epsilon_i > 9\epsilon_{i+1}$, since $\epsilon_{i+1} \geq \epsilon_n > b$. Therefore, the largest edge's weight is $4\epsilon_1$ and the smallest one is $\epsilon_n$ and their ratio is grater than $4\cdot9^{n-1}$.

Consider now the following initial opinion profile: players $B_1$ and $D_1$ have opinion $1/2$, $F_i = 1/2, \forall i=1,\ldots, n$, all the renaming players have opinion $0$. Basically. $G_1$ is the starting configuration of a switch off cycle.
Since $A_i$, for $i = 1, \ldots, n$, act as a switch, when her opinion switches from $0$ to $1/2$ ($1/2$ to $0$), a switch-off-cycle (switch-on) is executed on $G_{i+1}$. Note that the last two cases occur two times during the switch-off cycle of $G_i$.
An exponentially long opinion dynamics can start by switching-off $G_1$. Hence, $G_2$  goes through 2 switch-on cycles and 2 switch-off cycles, $G_3$ goes through 4 switch-on cycles and 4 switch-off cycles. Consequently, $G_n$ goes through $2^{n-1}$ switch-on cycles and $2^{n-1}$ switch-off cycles.
\end{proof}

\subsection{Proof of Theorem~\ref{conv_bound}}
In order to prove Theorem~\ref{conv_bound}, let us first consider the simpler case in which both the strength of social media's influence ($b$) and the weight of the edges are integer.
\begin{proposition}
\label{prop:conv_int}
Given a social network $G = (V,E)$ with $\|V\| = n$ and an opinion set $\Theta$ with discretization factor $\delta$, if both the strength $b$ of social media influence and the weights $(w_e)_{e \in E}$ of the edges are integers, then the opinion dynamics with asynchronous updates converges to a Nash Equilibrium in $O(4b\frac{n}{\delta^2}+4w_{max}\frac{n^2}{\delta^2})$, where $w_{\max} = \max_{e \in 	E} w_e$.
\end{proposition}
\begin{proof}
In the proof of Theorem \ref{thm:stable} we proved that $\Phi(x) = \sum_{i \in V} b(x_i - s(x_i))^2 + \sum_{(i,j) \in E} w_{i,j}(x_i-x_j)^2$ is a generalized ordinal potential function for the opinion dynamics game proposed. Trivially, we have $\Phi(x) \leq 4bn+4w_{max}n^2$, where $n$ is the number of player and it is equal to $\vert V \vert$ and $w_{max} = \max_{e \in E} w_e$.\\
Moreover, in the proof of Theorem \ref{thm:stable} we also proved that
\begin{equation*}
\label{diff_in_Phi}
\Phi(x) - \Phi(y, x_{-i}) = b[(x_i-s(x_i))^2-(y_i-s(y_i))^2]+\Delta P
\end{equation*}
where $\Delta P = P(x) - P(y_i, x_{-i}) = P(x) - P(y_i, x_{-i})$.\\
If a player $i$ change its opinion at step $t$, the following inequalities it is satisfied:
\begin{equation*}
c_i(x^t) - c_i(x_{-i}^t, y_i) = b(x_i^t - s(x_i^t))^2 - b(y_i - s(x_i^t))^2+\Delta P > 0
\end{equation*}
where $y_i$ is the best-response of player $i$ at step $t$.\\
Let $\Delta B = b(x_i^t - s(x_i^t))^2 - b(y_i - s(x_i^t))^2$, $(x_i^t - s(x_i^t)) = k\delta$ with $k \in \mathbb{N}$, therefore, knowing that $b \in \mathbb{N}$, it's easy to see why $b(x_i^t - s(x_i^t))^2 = k'\delta^2$, with $k' \in \mathbb{N}$. Trivially, we have  the same only for $b(y_i - s(x_i^t))^2$, hence $\Delta B = K\delta^2$, with $K \in \mathbb{Z}$.
In a very similar way, knowing that $w_e \in \mathbb{N}, \forall e \in E$, we can show that $\Delta P = K'\delta^2$, with $K \in \mathbb{Z}$.\\
Consequently, $\Phi(x^t) - \Phi(y, x_{-i}^t) \geq c_i(x^t) - c_i(x_{-i}^t, y_i) = K'\delta^2 \geq \delta^2$.\\
Basically, If a player change it's opinion, the generalized potential function decreases at least by $\delta^2$. Therefore the number of steep needed . 
\end{proof}

We now extend the bound for the convergence of the opinion dynamics to games whose influences have bounded precision $k$. To this aim, let us first define the concept of \emph{best-response equivalent} \citep{goldberg, dyer}.
\begin{definition}
Two best-response games $B, B'$ are best-response equivalent if they have the same set of players and strategies, and for any player and strategy profile, the best-response of that player is the same in $B$ as in $B'$.
\end{definition}
Let $B_{OP}$ the best-response game described in section \ref{secTheModel}, $B'_{OP}$ is exactly defined as $B_{OP}$ but for any player $i$ the cost is $c'_i(\cdot) = 10^kc_i(\cdot)$, where $c_i(\cdot)$ is the cost for the player $i$ in $B_{OP}$.
\begin{observation}
\label{br_eq_ob}
$B_{OP}$ and $B'_{OP}$ are best-response equivalent. $\Phi'(\cdot) = 10^k\Phi(\cdot)$ is the generalized ordinal potential function for $B'_{OP}$, where  $\Phi(\cdot)$ is the generalized ordinal potential function for $B_{OP}$.
\end{observation}
Now we are ready for proving Theorem~\ref{conv_bound} through an opportune extension of Proposition~\ref{prop:conv_int}.
\begin{proof}[Proof of Theorem~\ref{conv_bound}]
Let $B_{OP}$ the best-response opinion game whose edges' weights and social media's influence have bounded precision $k$, $B'_{OP}$ is exactly the same as $B_{OP}$, but the social media's influence and edges' weights are $10^k$ times greater than $B_{OP}$: $b' = 10^kb$, $w'_{e} = 10^kw_{e}, \forall e \in E$. Trivially, the cost of any player $i$ in $B'_{OP}$ is $10^k$ times the cost of that player in $B_{OP}$, $c'_i(\cdot) = 10^kc_i(\cdot)$. Therefore, from observation \ref{br_eq_ob} we know that $B_{OP}$ and $B'_{OP}$ are best-response equivalent. The influences of opinion game $B'_{OP}$ are integer, hence from Proposition~\ref{prop:conv_int}, it converges in $O(4b'\frac{n^2}{\delta^2}+4w'_{max}\frac{n^3}{\delta^2})$. Consequently $B'_{OP}$ converges in $O(10^k4b\frac{n^2}{\delta^2}+10^k4w_{max}\frac{n^3}{\delta^2})$.
\end{proof}

% \end{appendices}

\bibliographystyle{unsrtnat}
\bibliography{arxiv}

\end{document}